\documentclass[envcountsame]{llncs}

\usepackage{ifthen}

\newif\iffullversion
\fullversiontrue

\newif\ifdraft 
\drafttrue

\usepackage{color}
\usepackage[dvipsnames]{xcolor}

\usepackage{graphicx}
\usepackage[hidelinks,colorlinks=true,linkcolor=Blue,urlcolor=Blue,citecolor=OliveGreen,bookmarksdepth=3,bookmarksopen=true]{hyperref}
\usepackage[lambda]{cryptocode}
\usepackage{braket}
\usepackage{caption}
\usepackage{amsfonts}
\usepackage{amsmath}
\usepackage{xspace}
\usepackage{tikz}
\usepackage{mathtools}

\usepackage{marginnote}

\usepackage{ifthen}

\usepackage[noadjust]{cite}

\usepackage[normalem]{ulem}

\usepackage[caption=true]{subfig}
\usepackage{tikz}
\usetikzlibrary{chains,fadings,calc,shapes,shapes.geometric,shapes.misc,backgrounds,decorations,decorations.pathmorphing,decorations.pathreplacing,patterns,fit}
\usetikzlibrary{positioning}
\usetikzlibrary{shapes.geometric}
\usetikzlibrary{shadows}

\usepackage{algorithm, algpseudocode, float}

\usepackage[utf8]{inputenc}

\usepackage{eso-pic}

\ifdraft

\newcommand{\old}[1]{{\color{RawSienna}\sout{#1}}}

\newcommand{\todo}[1]{{\color{CornflowerBlue}{\textsf{\textbf{TODO:} #1}}}}
\newcommand{\task}[2]{{\color{CornflowerBlue}{\textsf{\textbf{TODO #1:} #2}}}}
\newcommand{\jk}[1]{{\color{RoyalPurple}{\textsf{\textbf{jk:} #1}}}}
\newcommand{\ps}[1]{{\color{BurntOrange}{\textsf{\textbf{ps:} #1}}}}
\newcommand{\tg}[1]{{\color{red}{\textsf{\textbf{tg:} #1}}}}

\else

\newcommand{\old}[1]{}

\newcommand{\todo}[1]{}
\newcommand{\task}[2]{}
\newcommand{\jk}[1]{}
\newcommand{\ps}[1]{}
\newcommand{\tg}[1]{}
\fi

\newcommand{\mathcmd}[1]{\ensuremath{#1}\xspace} 
\newcommand{\sdash}{\mbox{-}}

\newtheorem{exper}[theorem]{Experiment}

\newcommand{\algfont}[1]{\mathtt{#1}}  
\newcommand{\varfont}[1]{\mathit{#1}}    
\newcommand{\notionfont}[1]{\mathrm{#1}} 
\newcommand{\setfont}[1]{{\mathcal{#1}}} 
\newcommand{\orafont}[1]{{\mathsf{#1}}} 
\newcommand{\gamefont}[1]{{\mathsf{#1}}} 
\newcommand{\distributionfont}[1]{\mathcal{#1}} 
\newcommand{\gmfont}[1]{{\mathsf{#1}}}

\newcommand{\mssg}[1][]{\mathcmd{\if!#1! \varfont{m}\else \varfont{m}_{#1}\fi}} 
\newcommand{\altmssg}[1][]{\mathcmd{\if!#1! \varfont{m'}\else \varfont{m'}_{#1}\fi}} 
\newcommand{\ptxt}[1][]{\mathcmd{\if!#1! \varfont{m}\else \varfont{m}_{#1}\fi}} 
\newcommand{\ctxt}[1][]{\mathcmd{\if!#1! \varfont{c}\else \varfont{c}_{#1}\fi}} 
\newcommand{\altctxt}[1][]{\mathcmd{\if!#1! \varfont{c'}\else \varfont{c'}_{#1}\fi}} 
\newcommand{\rnd}[1][]{\mathcmd{\if!#1! \varfont{r}\else \varfont{r}_{#1}\fi}} 
\newcommand{\pk}{\ensuremath{\mathsf{pk}}\xspace} 
\newcommand{\sk}{\ensuremath{\mathsf{sk}}\xspace} 
\newcommand{\key}[1][]{\mathcmd{\if!#1! \mathsf{k}\else \mathsf{k}_{#1}\fi}} 

\renewcommand{\vec}[2][]{\mathcmd{\if!#1! \mathbf{#2}\else \mathbf{#2}_{#1}\fi}} 
\newcommand{\mat}[2][]{\mathcmd{\if!#1! \mathbf{#2}\else \mathbf{#2}_{#1}\fi}} 
\newcommand{\poly}[2][]{\mathcmd{\if!#1! \mathsf{#2}\else \mathsf{#2}_{#1}\fi}} 

\newcommand{\tTab}[1][]{\mathcmd{\if!#1! \varfont{t} \else \varfont{t}_{#1}\fi}} 

\newcommand{\query}[1][]{\mathcmd{\if!#1! \varfont{q} \else \varfont{q}_{#1}\fi}}

\let\oldket\ket
\let\oldbra\bra
\renewcommand{\ket}[1]{\mathcmd{\oldket{#1}}}
\renewcommand{\bra}[1]{\mathcmd{\oldbra{#1}}}

\newcommand{\func}{\mathcmd{\mathcal{F}}}

\newcommand{\roInp}[1][]{\mathcmd{\if!#1! \varfont{x}\else \varfont{x}_{#1}\fi}} 
\newcommand{\roOut}[1][]{\mathcmd{\if!#1! \varfont{y}\else \varfont{y}_{#1}\fi}} 
\newcommand{\AltroOut}[1][]{\mathcmd{\if!#1! \varfont{\bar{y}}\else \varfont{\bar{y}}_{#1}\fi}} 

\newcommand{\decisionalProblemChallenge}[1][]{\mathcmd{\if !#1!\varfont{c} \else \varfont{c}_{#1}\fi}} 

\renewcommand{\phi}{\mathcmd{\varphi}}
\renewcommand{\epsilon}{\mathcmd{\varepsilon}}

\renewcommand{\secpar}{\mathcmd{\lambda}}
\renewcommand{\secparam}{\mathcmd{1^\lambda}}

\newcommand{\win}{\mathcmd{\varfont{win}}}
\newcommand{\rej}{\mathcmd{\varfont{rej}}}

\newcommand{\norm}[2][]{\mathcmd{\if!#1! \left\lVert#2\right\rVert \else \left\lVert#2\right\rVert_{#1}\fi}} 
\newcommand{\tensor}{\otimes} 
\newcommand{\xor}{\oplus} 
\renewcommand{\gets}{\mathcmd{\coloneqq}}
\newcommand{\getsr}{\mathcmd{\leftarrow}} 
\newcommand{\from}{\mathcmd{\leftarrow}}
\newcommand{\fromunif}{\mathcmd{\raisebox{-1pt}{\ensuremath{\,\xleftarrow{\raisebox{-1pt}{$\scriptscriptstyle\$$}}\,}}}} 

\newcommand{\ceil}[1]{\mathcmd{\left\lceil #1 \right\rceil}}

\newcommand{\ketbra}[1]{\mathcmd{\ket{#1}\!\!\bra{#1}}}

\newcommand\NN{\mathbb{N}} 
\newcommand{\mssgsp}{\mathcmd{\setfont{M}}} 
\newcommand{\ctxtsp}{\mathcmd{\setfont{C}}} 
\newcommand{\bit}[1][]{\mathcmd{\if!#1!\{0,1\}\else\{0,1\}^{#1}\fi}} 
\newcommand{\distribution}[2][]{\mathcmd{\if!#1! \distributionfont{D}_{#2}\else \distributionfont{D}^{#1}_{#2} \fi}}

\newcommand{\keysp}{\mathcmd{\setfont{K}}}
\newcommand{\msgsp}{\mathcmd{\setfont{M}}}
\newcommand{\ctxsp}{\mathcmd{\setfont{C}}}
\newcommand{\rndsp}{\mathcmd{\setfont{R}}}

\newcommand{\sksp}{\mathcmd{\setfont{S}}}
\newcommand{\pksp}{\mathcmd{\setfont{P}}}
\newcommand{\states}{\mathcmd{\mathfrak D}}
\newcommand{\hilbert}{\mathcmd{\mathfrak H}}

\newcommand{\Eve}[1][]{\mathcmd{\if!#1!\mathcal{A}\else\mathcal{A}_{#1}\fi}}
\newcommand{\AltEve}[1][]{\mathcmd{\if!#1!\overline{\mathcal{A}}\else\overline{\mathcal{A}}_{#1}\fi}}
\newcommand{\EveB}[1][]{\mathcmd{\if!#1!\mathcal{B}\else\mathcal{B}_{#1}\fi}} 
\newcommand{\EveC}[1][]{\mathcmd{\if!#1!\mathcal{C}\else\mathcal{C}_{#1}\fi}} 
\newcommand{\EveR}[1][]{\mathcmd{\if!#1!\mathcal{R}\else\mathcal{R}_{#1}\fi}} 
\newcommand{\kgen}[1][]{\mathcmd{\if!#1! \algfont{KGen}\else \algfont{KGen}^{#1}\fi}} 
\newcommand{\enc}[1][]{\mathcmd{\if!#1! \algfont{Enc}\else \algfont{Enc}^{#1}\fi}} 
\newcommand{\dec}[1][]{\mathcmd{\if!#1! \algfont{Dec}\else \algfont{Dec}^{#1}\fi}} 
\newcommand{\rec}[1][]{\mathcmd{\if!#1! \algfont{Rec}\else \algfont{Rec}^{#1}\fi}} 
\newcommand{\Enc}{\mathcmd{\algfont{Enc}}} 
\newcommand{\Dec}{\mathcmd{\algfont{Dec}}} 
\newcommand{\KGen}{\mathcmd{\algfont{KGen}}} 
\newcommand{\Kgen}{\mathcmd{\algfont{KGen}}} 
\newcommand{\Rec}{\mathcmd{\algfont{Rec}}} 

\newcommand{\PRF}{\mathcmd{\algfont{F}}}
\newcommand{\TDF}{\mathcmd{\algfont{F}}}
\newcommand{\TDFinv}{\mathcmd{\algfont{F}^{-1}}}

\newcommand{\EveM}[1][]{\mathcmd{\if!#1!\mathcal{M}\else\mathcal{M}_{#1}\fi}}
\newcommand{\AltEveM}[1][]{\mathcmd{\if!#1!\overline{\mathcal{M}}\else\overline{\mathcal{M}}_{#1}\fi}}
\newcommand{\EveD}[1][]{\mathcmd{\if!#1!\mathcal{D}\else\mathcal{D}_{#1}\fi}}
\newcommand{\AltEveD}[1][]{\mathcmd{\if!#1!\overline{\mathcal{D}}\else\overline{\mathcal{D}}_{#1}\fi}}

\newcommand{\Altkgen}[1][]{\mathcmd{\if!#1! \overline{\algfont{KGen}}\else \overline{\algfont{KGen}}^{#1}\fi}} 
\newcommand{\Altenc}[1][]{\mathcmd{\if!#1! \overline{\algfont{Enc}}\else \overline{\algfont{Enc}}^{#1}\fi}} 

\newcommand{\adver}{\mathcmd{\mathcal{A}}} 
\newcommand{\chal}{\mathcmd{\mathcal{C}}} 

\newcommand{\game}[1][]{\mathcmd{\if!#1! \gmfont{G}\else \gmfont{G}_{#1}\fi}} 
\newcommand{\gamebox}[1][]{\setlength{\fboxsep}{0.5mm}\fbox{\mathcmd{\if!#1! \gmfont{G}\else \gmfont{G}_{#1}\fi}}} 
\newcommand{\hyb}[1][]{\mathcmd{\if!#1! \gmfont{H}\else \gmfont{H}_{#1}\fi}} 
\newcommand{\hybbox}[1][]{\mathcmd{\if!#1! \setlength{\fboxsep}{0.1mm}\fbox{\mathcmd{\gmfont{H}}}\else \mathcmd{\setlength{\fboxsep}{0.1mm}\fbox{\mathcmd{\gmfont{H}}}_{#1}}\fi}} 

\newcommand{\qINDqCPA}{\mathcmd{\notionfont{qIND{\sdash}qCPA}}} 
\newcommand{\piqINDqCPA}{\mathcmd{\notionfont{\pi{\sdash}qIND{\sdash}qCPA}}} 

\newcommand{\QIND}{\mathcmd{\notionfont{QIND}}} 

\newcommand{\INDqCPA}{\mathcmd{\notionfont{IND{\sdash}qCPA}}} 
\newcommand{\fqINDCPA}{\mathcmd{\notionfont{fqIND{\sdash}CPA}}} 
\newcommand{\qIND}{\mathcmd{\notionfont{qIND}}} 
\newcommand{\qCPA}{\mathcmd{\notionfont{qCPA}}} 

\newcommand{\CPA}{\mathcmd{\notionfont{CPA}}}
\newcommand{\CCA}{\mathcmd{\notionfont{CCA}}}

\newcommand{\CCAtwo}{\mathcmd{\notionfont{CCA2}}}
\newcommand{\qINDqCCAone}{\mathcmd{\notionfont{qIND{\sdash}qCCA1}}} 
\newcommand{\qINDqCCAtwo}{\mathcmd{\notionfont{qIND{\sdash}qCCA2}}} 

\newcommand{\Adv}[3][]{\mathcmd{\if!#1! \mathbf{Adv}^{#2}(#3) \else \mathbf{Adv}_{#1}^{#2}(#3) \fi}} 
\newcommand{\gameAdv}[2]{\mathcmd{\mathbf{Adv}\left(#1,#2\right)}} 
\newcommand{\INDCPA}{\mathcmd{\notionfont{IND{\sdash}CPA}}} 
\newcommand{\QPRP}{\mathcmd{\notionfont{QPRP}}} 

\newcommand{\succProb}[1][]{\mathcmd{\if!#1! \notionfont{\epsilon} \else \notionfont{\epsilon}_{#1}\fi}}

\newcommand{\negl}{\mathcmd{\notionfont{negl}}}

\newcommand{\liftableGame}[1][]{\mathcmd{\if!#1! \gmfont{\widetilde{G}}\else \gmfont{\widetilde{G}}_{#1}\fi}} 

\newcommand{\cGame}[1][]{\mathcmd{\if!#1! \gmfont{\bar{\game}}\else \gmfont{\bar{\game}}_{#1}\fi}} 
\newcommand{\pqGame}[1][]{\mathcmd{\if!#1! \gmfont{\widetilde{\game}}\else \gmfont{\widetilde{\game}}_{#1}\fi}} 
\newcommand{\cGamebox}[1][]{\mathcmd{\if!#1! \gmfont{\setlength{\fboxsep}{0.5mm}\fbox{\cGame}}\else \gmfont{\setlength{\fboxsep}{0.5mm}\fbox{\cGame[#1]}}\fi}} 
\newcommand{\pqGamebox}[1][]{\mathcmd{\if!#1! \gmfont{\setlength{\fboxsep}{0.5mm}\fbox{\pqGame}}\else \gmfont{\setlength{\fboxsep}{0.5mm}\fbox{\pqGame[#1]}}\fi}} 

\newcommand{\hashOracle}[1][]{\mathcmd{\if!#1! \orafont{H}\else \orafont{H}_{#1}\fi}} 

\newcommand{\randomOracle}{\mathcmd{\orafont{O}}} 

\newcommand{\x}{\enspace}

\newcommand{\pcoracle}{\highlightkeyword{Oracle}}

\newcommand{\pcas}{\highlightkeyword{as}}

\newcommand{\scaleFactorGraphics}{1} 

\title{Quantum Indistinguishability\\for Public Key Encryption}
\author{Tommaso Gagliardoni\inst{1} \and Juliane Krämer\inst{2} \and Patrick Struck\inst{2}}
\institute{
	Kudelski Security, Switzerland\\
	\email{firstname.lastname@kudelskisecurity.com}
	\and Technische Universität Darmstadt, Germany\\
	\email{firstname@qpc.tu-darmstadt.de}
}

\iffullversion
\AddToShipoutPicture*{\centering\small
	\raisebox{26.0cm}{\hspace{4.75cm}\parbox{\textwidth}{
	\begin{normalsize}
	An extended abstract of this work appears (with the same title) in the proceedings of \emph{The Twelfth International Conference on Post-Quantum Cryptography (PQCrypto 2021)}. This is the full version.
	\end{normalsize}
	}
	}
	}
\fi

\begin{document}
	
\maketitle
\noindent

\begin{center}
\today
\end{center}

\begin{abstract}
	In this work we study the quantum security of public key encryption schemes (PKE).
	Boneh and Zhandry (CRYPTO'13) initiated this research area for PKE and symmetric key encryption (SKE), albeit restricted to a classical indistinguishability phase.
	Gagliardoni et al. (CRYPTO'16) advanced the study of quantum security by giving, for SKE, the first definition with a quantum indistinguishability phase.
	For PKE, on the other hand, no notion of quantum security with a quantum indistinguishability phase exists.
	
	Our main result is a novel quantum security notion (\qINDqCPA) for PKE with a quantum indistinguishability phase, which closes the aforementioned gap.
	We show a distinguishing attack against code-based schemes and against LWE-based schemes with certain parameters.
	We also show that the canonical hybrid PKE-SKE encryption construction is \qINDqCPA-secure, even if the underlying PKE scheme by itself is not.
	Finally, we classify quantum-resistant PKE schemes based on the applicability of our security notion.
	
	Our core idea follows the approach of Gagliardoni et al. by using so-called \mbox{type-2} operators for encrypting the challenge message.
	At first glance, \mbox{type-2} operators appear unnatural for PKE, as the canonical way of building them requires both the secret and the public key.
	However, we identify a class of PKE schemes - which we call \emph{recoverable} - and show that for this class \mbox{type-2} operators require merely the public key.
	Moreover, recoverable schemes allow to realise \mbox{type-2} operators even if they suffer from decryption failures, which in general thwarts the reversibility mandated by \mbox{type-2} operators.
	Our work reveals that many real-world quantum-resistant PKE schemes, including most NIST PQC candidates and the canonical hybrid construction, are indeed recoverable.
\end{abstract}


\newpage

\section{Introduction}

The discovery of Shor's~\cite{FOCS:Shor94} and Grover's~\cite{STOC:Grover96} quantum algorithms had a significant 
impact on cryptographic research.
Shor's algorithm in particular has the potential to completely break most of the public key cryptosystems used nowadays.
This led to the development of quantum-resistant cryptography,\footnote{This type of cryptography is often called ``post-quantum cryptography''~\cite{bernsteinpost}.}
that is, cryptography that can run on non-quantum computers but should withstand attackers equipped with quantum computing power. 
In recent years the research efforts on quantum-resistant cryptography accelerated significantly due to the standardisation process initiated by the NIST~\cite{NISTstandardization}.

Modern cryptography is based on the paradigm of {\em provable security}, which is itself given in terms of a security notion, an adversarial model, and a security proof. A widely used framework for defining security notions is the so-called {\em game-based security}, which is presented as a game between two or more parties.\footnote{Other frameworks exist, such as simulation-based, but as a first approximation game-based security notions are very convenient for their intuitivity and simplicity.}
In the case of encryption schemes these parties are: 
a challenger, representing the user of the scheme, and an adversary, representing an attacker against the scheme.
Any meaningful model for quantum-resistant schemes should entail that the adversary has quantum computing power.
Based on this, we can differentiate between different models depending on the computing power of the challenger. 
In the literature there are mainly two of these models that are taken into account. 
In the first, the challenger remains fully classical, implying that any communication between adversary and challenger is also classical (including oracles provided by the challenger to the adversary), while the adversary retains local quantum computing power. 
This is the model most often considered in quantum-resistant cryptography, and it is also called QS1~\cite{DBLP:phd/dnb/Gagliardoni17} or Q1~\cite{ToSC:KLLN16}. 
In the second case, the challenger also has quantum computing power, which enables quantum communication between challenger and adversary. 
This stronger model is sometimes called ``superposition-attack security''~\cite{ICITS:DFNS13}, QS2~\cite{DBLP:phd/dnb/Gagliardoni17}, or Q2~\cite{ToSC:KLLN16}.

Boneh and Zhandry~\cite{C:BonZha13} initiated the study of QS2 security for cryptographic primitives.
For signature schemes, they give a security definition that allows the adversary to query the signing oracle on a superposition of messages.
For public and symmetric key encryption (PKE and SKE) schemes, on the other hand, they prove that simply allowing the adversary to query a superposition of messages as challenge in a ``natural'' way gives an unachievable security notion (\fqINDCPA).
This is due to entanglement between the plaintext register and the ciphertext register.
They show how to exploit this entanglement to break this security notion irrespectively of the used encryption scheme.
To resolve this, they propose another security notion (\INDqCPA) which allows the adversary superposition queries in the CPA phase while the challenge messages in the IND phase are restricted to be classical.  This notion coincides with the traditional QS1 security notion for PKE (as the adversary can simulate the encryption in superposition using his local computing power and the public key), while for SKE, this yields a notion of QS2 security -  although the restriction to a classical challenge in this case is clearly a shortcoming.

Gagliardoni et al.~\cite{C:GagHulSch16} overcame this shortcoming in the symmetric key case by showing how to model a quantum challenge query, while keeping the resulting security notion (\qINDqCPA) still achievable, yet stronger than \INDqCPA.
At the heart of their idea lies the use of so-called {\em type-2 operators\footnote{Also called \emph{minimal oracles} in~\cite{kashefi2002comparison}.}} rather than so-called type-1 operators when encrypting the challenge messages of the adversary.
Type-1 operators are the ``canonical'' way of implementing a classical function \func on a quantum superposition of input,  by mapping  the state $\ket{x,y}$ to $\ket{x,y\xor\func(x)}$, thereby ensuring reversibility for any function \func (reversibility being necessary when defining non-measurement quantum operations).
An important property of type-1 operators is that they create entanglement between the input and output registers.
This is exactly the entanglement which Boneh and Zhandry exploit to show that \fqINDCPA is unachievable.
In contrast to these, type-2 operators work directly on the input register, i.e., they map the state $\ket{x}$ to $\ket{\func(x)}$.
Only reversible functions, for instance permutations, can be implemented as type-2 operators, while it is impossible to compute, say, an arbitrary one-way function through a type-2 operator.
Gagliardoni et al. observe that SKE schemes act as permutations between the plaintext space and the ciphertext space, which allows to implement the encryption algorithm as a type-2 operator. 
This, in turn, allows to build a solid framework for QS2 security in the case of SKE.

In~\cite{C:GagHulSch16} the authors speculate that their techniques could be extended to the public key case (PKE) as well. 
However, defining type-2 operators for PKE schemes is much more involved than for SKE schemes.
First, to achieve \INDCPA security, PKE schemes are inherently randomised and the randomness is usually erased in the 
process of decryption.
Second, many constructions for quantum-resistant PKE schemes, in particular lattice-based and code-based schemes, suffer from a small probability of decryption failures, i.e., ciphertexts which do not decrypt correctly.
Given the above, at first glance it is unclear whether type-2 operators for PKE schemes are possible at all, as these 
two properties seem to thwart the mandatory reversibility. 
Hence, QS2 security for the public key case remained an open problem so far.

\subsection{Our Contribution}

We present a novel QS2 security notion\footnote{See Appendix~\ref{app:concurrentWork} for independent and concurrent work.} for PKE, provide both achievability results and separation to the QS1 security notion
for many real-world schemes, 
and give a general classification of PKE with respect to our security notion.

Our core focus is to extend the results from~\cite{C:GagHulSch16} to the public key scenario. 
We first formalise the theory of type-2 encryption operators for PKE. 
For perfectly correct schemes (i.e., schemes which do not suffer from the possibility of decryption failures) we define 
the type-2 operator to preserve a randomness register in input and output. 
Even if such approach might look strange at first glance, we show that this is the most natural way of defining type-2 
operators for PKE schemes. 
As a next step, we identify a class of PKE schemes (which we call {\em recoverable}) where decryption failures can 
always be avoided given knowledge of the randomness used during encryption, regardless of the actual failure probability 
of the decryption algorithm. 
We observe that most real-world partially correct PKE schemes (including many quantum-resistant NIST candidates) are 
actually of this type.
Then, for schemes that are perfectly correct or recoverable, we show how to efficiently construct the type-2 encryption 
operator.
Moreover, we show that for recoverable schemes, this can be done by knowledge of the public key only!
This implies, perhaps surprisingly, that the adversary can implement efficiently this type-2 operator already in the QS1 
model.
Such observation marks a substantial difference from the symmetric key case, where the need for type-2 operators is 
dictated by necessity in order to cover exotic attack models.

Using the theory of type-2 operators developed so far, we give a novel QS2 security notion for PKE, that we call \emph{quantum ciphertext indistinguishability under quantum chosen plaintext attack} (\qINDqCPA). 
For a new security notion to be meaningful, two properties are required.
First, it has to be achievable, and, second, it has to differ from existing security notions.

We analyse several real-world PKE schemes in respect to our new \qINDqCPA security notion.
We show that the canonical LWE-based PKE scheme~\cite{STOC:Regev05} can be attacked, at least for certain parameters.
The attack is similar to the ``Hadamard distinguisher'' given in~\cite{C:GagHulSch16}.
Moving on to code-based schemes, we observe that some constructions encrypt the message using a one-time pad operation, which again allows to exploit the distinguishing attack.
As an example we show that the code-based scheme ROLLO-II~\cite{NISTPQC-R2:ROLLO19} is not \qINDqCPA secure.

However, in practice most real-world PKE schemes (including the NIST submissions) are used as Key Encapsulation Mechanisms (KEM) in combination with an SKE scheme, yielding a hybrid PKE-SKE construction. 
Looking at such canonical hybrid construction then, we show that its \qINDqCPA security 
mostly depends on the underlying SKE scheme, while the PKE scheme only needs to be secure in the QS1 sense.
Hence, even the code-based PKE scheme ROLLO-II, which as a stand-alone PKE scheme is not \qINDqCPA-secure, can be used to achieve \qINDqCPA security if combined with a \qINDqCPA-secure SKE scheme via the hybrid construction, which is the default way of using it in practice.

We additionally discuss the difficulty of defining type-2 operators (and the related QS2 security notion) for arbitrary schemes that are neither perfectly correct nor recoverable.
For this, we study the problem of their general classification and we identify a class of schemes, that we call {\em isometric}, that allow to overcome such difficulty.
Furthermore, we provide constructions and separation results.

\subsection{The Motivation for QS2 Security} \label{sec:MotivationQS2}

Defining security against quantum adversaries with superposition access to certain oracles requires some motivation. 
Sometimes, the resulting security notion is already implicitly captured by the corresponding QS1 scenario (for example in the case of {\em quantum random oracles}~\cite{AC:BDFLSZ11}).
In other cases, for instance those considered in~\cite{DBLP:conf/isita/KuwakadoM12,EC:AlaRus17,C:KLLN16}, it might look like an artificial extension of the theory.

However, QS2 security extends quantum properties to types of attack scenarios not covered in QS1, and at the same time ``bridges'' certain security notions from the classical realm to schemes which are meant to run natively on a quantum computer.
Some of the reasons why QS2 notions are important to consider are explained in detail in~\cite{DBLP:phd/dnb/Gagliardoni17}.
They basically boil down to five points.
\begin{enumerate}
	\item To ensure that quantum-resistant classical schemes retain their security even if executed on a quantum computer, possibly in complex environments or protocols where composition should be taken into account.
	\item To fix security proofs, where the sole QS1 security of certain underlying building blocks is not enough to ensure that the whole proof goes through.
	An example is the need of QS2-secure pseudorandom functions (QPRF) in order to simulate a quantum random oracle~\cite{FOCS:Zhandry12}, which is a QS1 concept.
	\item To ensure the security of quantum protocols (i.e., meant to run natively on a quantum computer and protect quantum data) when using classical algorithms as building blocks.
	For example,~\cite{DBLP:phd/dnb/Gagliardoni17} shows how it is possible to build a secure symmetric quantum encryption scheme (falling into the so-called QS3 domain) by using a \qINDqCPA symmetric 
	classical encryption scheme (QS2), but not necessarily a simple quantum-resistant (QS1) one.
	\item To consider cases of {\em code obfuscation}; for example creating a quantum-resistant PKE scheme by hardcoding a  
	symmetric key into an obfuscated encryption program (a technique known as {\em whiteboxing}~\cite{SAC:CEJV02}), which 
	is then distributed as a public key.
	\item To cover cases of {\em exotic quantum attacks}.
	These include, for instance, {\em quantum fault injection attacks}, where a classical device is subject to 
	controlled and artificial physical conditions that induce full or partial quantum behaviour of its hardware 
	(``tricking'' a classical device into being quantum, like in the ``frozen smart-card attack'' presented 
	in~\cite{C:GagHulSch16}); 
	or cases where a quantum computer is used to run a classical algorithm, but an adversary manages to intercept the intermediate result of the computation {\em before} the final measurement meant to produce a classical outcome.
\end{enumerate}
In our specific case, our results follow from the core use of type-2 operators.
This kind of quantum operations is poorly studied in the quantum computing realm, and might therefore look artificial 
for cryptographic use. 
In the present work we make an effort to expand in a detailed way the formalisation of such operators which, we stress, are only given for functions that are inherently invertible. 
It is a well-known fact (see for example~\cite{C:GagHulSch16}) that implementing these operators for encryption schemes usually requires knowledge of the secret key. 
We do not consider this to be a limitation because in the quantum setting, an honest challenger equipped with the secret key could be allowed to generate particular ciphertext-encoding states which would be hard to compute for an external party: it is therefore necessary to cover this distinction in the preparation of ciphertext states, and type-2 operators do just that. 
Moreover, as we show in the present work, for many natural PKE schemes, type-2 encryption operators can actually be efficiently implemented by knowledge of the public key only.

\subsection{Related Work}

The study of quantum security under adversarial queries in superposition can be traced back to works such as~\cite{ICITS:DFNS13, DBLP:journals/siamcomp/Watrous09, AC:BDFLSZ11}, which explore different settings where this additional adversarial power has an impact on security. 
However, for the case of signatures and encryption schemes, the first framework going beyond the traditional QS1 paradigm was given in~\cite{C:BonZha13}.
This paradigm was further extended in~\cite{C:GagHulSch16} for symmetric key encryption schemes, and in~\cite{EC:AMRS20} for MACs/signatures.

Regarding examples of exotic quantum attacks previously mentioned: it is currently not known whether any of these are feasible at all, but as noted in~\cite{DBLP:phd/dnb/Gagliardoni17}: (1) if they are feasible, in some cases they do not even require a fully fledged quantum computer (for example, in the attack from~\cite{C:GagHulSch16} it would be only necessary to produce and detect a Hadamard superposition of messages); and (2) it is already known in the literature that these attacks can be devastating. 
For example, related-key attacks~\cite{DBLP:journals/ipl/RottelerS15}, and superposition attacks against Even-Mansour~\cite{DBLP:conf/isita/KuwakadoM12}, Feistel networks~\cite{RSA:IHMSI19,DBLP:conf/isit/KuwakadoM10}, block ciphers~\cite{PQCRYPTO:ATTU16,EC:AlaRus17}, and HMAC constructions~\cite{C:KLLN16}.

Qualitatively different, but technically very connected to the QS2 setting is the {\em fully quantum setting}, or QS3 in short. 
This security domain encompasses security notions and constructions for schemes which are natively run on quantum hardware. 
In the case of QS3 encryption, these are schemes which are meant to protect quantum, rather than classical data. 
It turns out that many of the challenges in this area are shared with the QS2 case. 
In the computational security setting, the first security notions have been provided in~\cite{C:BroJef15} for the CPA case, and in~\cite{ICITS:ABFGSJ16} for the CCA1 and semantic security case. 
These results have been further extended to the CCA2 setting in~\cite{EC:AlaGagMaj18} for the symmetric case, and in~\cite{DBLP:journals/iacr/AlagicGM18} for the public key case.

In concurrent and independent work, Chevalier et al.~\cite{eprint:ChevalierEV20} propose alternative QS2 security notions for encryption schemes. Their and our notion are incomparable, as also claimed in a recent work by Carstens et al.~\cite{eprint:CETU20}. We discuss the differences in more detail in Appendix~\ref{app:concurrentWork}.

\subsection{Organization of the Paper}

Section~\ref{sec:background} gives the required background for this work.
In Section~\ref{sec:QuantumSecurityForPKE} we study type-2 operators for PKE schemes, 
define {\em recoverable} schemes, 
and give our new quantum security notion for those.
Positive and negative results for real-world PKE schemes are presented in Section~\ref{sec:SecurityAnalysis}.
Finally, we refine the classification of PKE schemes in terms of QS2 security in Section~\ref{sec:classif} and conclude with open questions in Section~\ref{sec:conclusions}.

\section{Preliminaries} \label{sec:background}

In the following, we use ``classical'' as meaning ``non-quantum''. 
By {\em algorithm} or {\em procedure} we mean a uniform family of circuits (classical or quantum) of depth and width polynomial in the index of the family. 
We call such index a {\em security parameter}, and we denote it by \secpar (or \secparam if written in unary notation). 
We implicitly assume that all algorithms take \secparam as a first input, so we will often omit this. 
If a classical algorithm $A$ is deterministic, we denote its output $y$ on input $x$ as $y \gets A(x)$, while if it is randomised we use $y \getsr A(x)$; 
when derandomising an algorithm we look at the deterministic algorithm obtained when considering explicitly the internal randomness $r$ as an additional auxiliary input, and we write $y \gets A(x;r)$. 
We will also use $x \getsr \mathcal{D}$ to denote that an element $x$ is sampled from a distribution $\mathcal{D}$; or we will write $x \fromunif X$ if $x$ is sampled uniformly at random from a set $X$. 
We will call {\em negligible} a function that grows more slowly than any inverse polynomial, and {\em overwhelming} a function which is $1$ minus a negligible function.
Finally, $a \| b$ denotes concatenation of $a$ and $b$.

\subsection{Quantum Notation}

We assume familiarity with the topic of quantum computing, but recall here the basic required notation.
For an in-depth discussion we refer to~\cite{Nielsen:2011:QCQ:1972505}.

A quantum system, identified by a letter $A$, is represented by a complex Hilbert space, which we denote by $\hilbert_{A}$.
If $A$ is clear from the context, we write \hilbert rather than $\hilbert_{A}$. 
Pure states in a Hilbert space \hilbert are representatives of equivalence classes of elements of \hilbert of norm $1$. 
Mixed states, on the other hand, are a more general representation of quantum states that takes {\em entanglement} with external systems into account; 
they are elements of the density matrix operator space over \hilbert, that is, Hermitian positive semi-definite linear operators of trace $1$, denoted as $\states(\hilbert)$. 
We use the ket notation for pure states, e.g., \ket{\phi}, while mixed states will be denoted by lowercase Greek letter, e.g., $\rho$.
Operations on pure states from $A$ to $B$ are performed by applying a unitary operator $U \colon \hilbert_A \rightarrow \hilbert_B$ to the state, 
while the more general case of operations on mixed states is described by superoperators of the form $U \colon \states(\hilbert_A) \rightarrow \states(\hilbert_B)$

The canonical way to compute a classical function $\func \colon \setfont{X} \to \setfont{Y}$ on a superposition of possible inputs $\sum_{x \in \setfont{X}} \alpha_x \ket{x}$ is through the so-called {\em type-1 operator for \func} described by:
\begin{align*}
	U^{(1)}_{\func} \colon \sum_{x,y} \alpha_{x,y} \ket{x,y} \mapsto \sum_{x,y} \alpha_{x,y} \ket{x,y\xor\func(x)}\,.
\end{align*}
This can always be implemented efficiently whenever \func is efficient~\cite{Nielsen:2011:QCQ:1972505}. 
By linearity, it is sufficient to specify just the behaviour on the basis elements, i.e.:
\begin{align*}
	U^{(1)}_{\func} \colon \ket{x,y} \mapsto \ket{x,y\xor\func(x)}\,.
\end{align*}
If \func is invertible, then there is another non-equivalent possible way to compute \func in superposition. 
This is done through the so-called {\em type-2 operators}, which are defined as the unitary:
\begin{align*}
	U^{(2)}_{\func} \colon \ket{x} \mapsto \ket{\func(x)}\,.
\end{align*}
See Fig.~\ref{fig:type1type2operators} for an illustration of these different operators.
Kashefi et al.~\cite{kashefi2002comparison} first introduced type-2 operators using the term \emph{minimal oracles} instead.
They show that these operators are strictly stronger by giving a problem which can be solved exponentially faster with type-2 operators than with type-1 operators.
They also observe that the adjoint of the type-2 operator corresponds to the type-2 operator of the inverse function $\func^{-1}$, which is (usually) not the case for type-1 operators.
Besides that, type-2 operators have been used by Gagliardoni et al.~\cite{C:GagHulSch16} to define quantum security for secret key encryption schemes.

\begin{figure}[htbp]
	\centering
	\scalebox{\scaleFactorGraphics}{
		\begin{tikzpicture}
			
			\node (MIDDLE) [text width=2em,align=center] {};
			
			\node (Utype1)[text width=2em,align=center,fill=gray!25,minimum height=2.0cm,minimum width=1cm,draw] [left of=MIDDLE, node distance = 3.5cm] {\large $U_{\func}^{(1)}$};
			\node (Utype2)[text width=2em,align=center,fill=gray!25,minimum height=2.0cm,minimum width=1cm,draw] [right of=MIDDLE, node distance = 3.5cm] {\large $U_{\func}^{(2)}$};
			\node (IN1) [text width=1em,align=center] [left of=Utype1, yshift=0.5cm, node distance=1.5cm] {\large \ket{x}};
			\node (IN2) [text width=1em,align=center] [left of=Utype1, yshift=-0.5cm, node distance=1.5cm] {\large \ket{y}};
			\node (OUT1) [text width=1em,align=center] [right of=Utype1, yshift=0.5cm, node distance=1.5cm] {\large \ket{x}};
			\node (OUT2) [text width=1em,align=center] [right of=Utype1, yshift=-0.5cm, node distance=1.5cm] {\large \ket{y\xor\func(x)}};
			
			\draw[] ($(IN1.east) + (0.1,0.0)$) -- ($(Utype1.west) + (0.0,0.5)$);
			\draw[] ($(IN2.east) + (0.1,0.0)$) -- ($(Utype1.west) + (0.0,-0.5)$);
			
			\draw[] ($(Utype1.east) + (0.0,0.5)$) -- ($(OUT1.west) + (0.0,0.0)$);
			\draw[] ($(Utype1.east) + (0.0,-0.5)$) -- ($(OUT2.west) + (0.0,0.0)$);
			
			\node (IN3) [text width=1em,align=center] [left of=Utype2, yshift=0.0cm, node distance=1.5cm] {\large \ket{x}};
			\node (OUT3) [text width=1em,align=center] [right of=Utype2, yshift=0.0cm, node distance=1.5cm] {\large \ket{\func(x)}};
			
			\draw[] ($(IN3.east) + (0.1,0.0)$) -- ($(Utype2.west) + (0.0,0.0)$);
			
			\draw[] ($(Utype2.east) + (0.0,0.0)$) -- ($(OUT3.west) + (0.0,0.0)$);
			
		\end{tikzpicture}
	}
	\caption{Type-1 operator (left) and type-2 operator (right) for a function \func.}
	\label{fig:type1type2operators}
\end{figure}
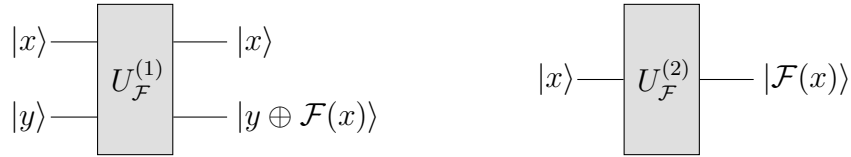

\subsection{Public Key Encryption}

In this section we give the formal definition for public key encryption schemes and the correctness of such schemes.

\begin{definition}\label{def:PKEscheme}
	A {\em public key encryption (PKE) scheme}  is a tuple $(\kgen,\enc,\dec)$ of three efficient algorithms such that:
	\begin{itemize}
		\item $\kgen \colon \NN \times \rndsp \rightarrow \pksp \times \sksp$ is the key generation algorithm which takes a security parameter \secpar and a randomness \rnd as input,  and returns a keypair consisting of a public key \pk and a secret key \sk. 
		If clear from the context, we will denote it by $(\pk,\sk) \getsr \KGen$.
		
		\item $\enc \colon \pksp \times \mssgsp \times \rndsp \rightarrow \ctxtsp$ is the encryption algorithm which takes a public key \pk, a message \mssg, and a randomness \rnd as input, and returns a ciphertext \ctxt. It will be usually denoted by $\ctxt \from \Enc_\pk(\mssg)$ or $\ctxt \gets \Enc_\pk(\mssg;\rnd)$.
		
		\item $\dec \colon \sksp \times \ctxtsp \rightarrow \mssgsp$ is the (deterministic) decryption algorithm\footnote{For simplicity here we only consider decryption with {\em implicit rejection}, that is, such the output is a random value whenever the input is not a well-formed ciphertext for the particular \sk.
			The extension to {\em explicit rejection} decryption can be done for example by adding a {\em flag bit} that marks the output as $\bot$ whenever decryption fails.}
		which takes as input a secret key \sk and a ciphertext \ctxt, and returns a message \mssg. It will be usually denoted by $\mssg \gets \Dec_\sk(\ctxt)$.
	\end{itemize}
	By \pksp, \sksp, \mssgsp, \ctxtsp, and \rndsp, we denote the public key space, secret key space, message space, ciphertext space, and randomness space, respectively.
\end{definition}
We assume w.l.o.g. that the randomness space for key generation and encryption are identical. 
Below we define two notions of correctness for PKE schemes. 
\begin{definition}[Perfectly Correct PKE]\label{def:perfectlyCorrectPKE}
	A PKE scheme $\Sigma = (\kgen,\enc,\dec)$ is \emph{perfectly correct} if for any $(\pk,\sk) \getsr \KGen$, $\mssg \in \mssgsp$, and $\rnd \in \rndsp$, 
	it holds that
	\begin{align*}
		\dec_\sk(\enc_\pk(\mssg;\rnd)) = \mssg\,.
	\end{align*}
\end{definition}
\begin{definition}[$(1-\alpha)$-Correct PKE~\cite{EC:DwoNaoRei04}]\label{def:partiallyCorrectPKE}
	A {\em $(1-\alpha)$-correct PKE scheme}, or {\em PKE with decryption error $\alpha$}, is a PKE scheme $\Sigma=(\kgen,\enc,\dec)$ such that, for any $\mssg \in \mssgsp$:
	\begin{align*}
		\underset{{\underset{r \getsr \rndsp}{(\pk,\sk) \getsr \KGen}}}{\Pr} \left[ \Dec_\sk(\Enc_\pk(\mssg;\rnd)) \neq \mssg \right] \leq \alpha\,.
	\end{align*}
\end{definition}

\section{Quantum Indistinguishability for PKE Schemes} \label{sec:QuantumSecurityForPKE}

In this section we extend the QS2 security notion of \qINDqCPA introduced for SKE schemes in~\cite{C:GagHulSch16} to the public key case. 
This is much more complex than the symmetric case, for the following reasons:
\begin{enumerate}
	\item PKE schemes are randomised to achieve ciphertext indistinguishability and adversaries must not learn the used randomness, even in the one-time case;
	\item when derandomising the encryption procedure and considering the randomness as additional input, there might be collisions (different randomnesses leading to the same ciphertext), hence ensuring reversibility of type-2 operators is not straightforward;
	\item many existing schemes, such as lattice- or code-based NIST candidates, suffer from a small decryption failure probability.
\end{enumerate}
In particular, as we will see, there are two main consequences: (1) the inverse of type-2 encryption operators is not generally a type-2 decryption operator; and (2), most interestingly, many type-2 encryption operators can be built efficiently by using only knowledge of the public key. 
The last point is crucial: it shows that in the PKE case, type-2 encryption operators are much more natural than in the SKE case, and for certain schemes they are actually already covered in the usual notion of QS1 ``post-quantum'' security. 
We will also show that some of these schemes are very relevant, such as the LWE-based scheme used as a blueprint for many NIST submissions.
In this section we will do the following:
\begin{enumerate}
	\item First, we revisit and define formally type-1 operators for PKE, and we show the difference between type-1 encryption and decryption (cf. Section~\ref{sec:type1PKE}).
	
	\item We define type-2 encryption operators for perfectly correct PKE schemes, and we show that they can be efficiently implemented with knowledge of secret and public key (cf. Section~\ref{sec:type2PKE}).
	
	\item We define what we call {\em recoverable} PKE schemes, i.e., schemes that admit an efficient procedure to recover the message given randomness, ciphertext and public key, without the secret key. We show that for such schemes the `canonical' type-2 encryption operator can be built by only using the public key, {\em even} if the scheme is not perfectly correct (cf. Section~\ref{sec:recoverablePKE}).
	
	\item We define the \qINDqCPA security notion for any PKE scheme where one can efficiently build the type-2 encryption operator. This includes in particular perfectly correct and recoverable schemes (cf. Section~\ref{sec:qINDqCPA}).
	
	\item Finally, we discuss how to extend these results to the {\em chosen ciphertext attack} (\CCA) scenario (cf. Section~\ref{sec:qINDqCCA}).
\end{enumerate}
%

\subsection{Type-1 Operators for PKE}\label{sec:type1PKE}

Recall that, for an arbitrary function $f: \setfont{X} \to \setfont{Y}$, the corresponding \mbox{type-1} operator is the ``canonical'' way of computing $f$ on a superposition of input through the unitary operator $U_f: \hilbert_\setfont{X} \otimes \hilbert_\setfont{Y} \to \hilbert_\setfont{X} \otimes \hilbert_\setfont{Y}$ defined by: $U_f : \ket{x,y} \mapsto \ket{x, y \xor f(x)}$. 
Realising $U_f$ is always efficient if $f$ is efficiently computable.

Traditionally, when looking at (deterministic) encryption schemes, the type-1 operator for encryption has been defined as:
\begin{equation*}
	U_{\enc} : \ket{\mssg,y} \mapsto \ket{\mssg, y \xor \enc(\mssg)}\,.
\end{equation*}
This is the approach used in, e.g.,~\cite{C:BonZha13} and~\cite{C:GagHulSch16}.
However, in our case of PKE schemes (which are generally randomised), we have to consider that encryption can be  
performed locally by the quantum adversary, who therefore has full control not only on the randomness used for 
encryption (i.e., it is necessary to explicitly derandomise the encryption procedure\footnote{This is implicitly 
	considered in \cite{C:BonZha13} and \cite{C:GagHulSch16}, but not explicitly formalised.}), but also on the public key 
used (i.e., it is theoretically possible to compute encryption for a superposition of different public keys).
Therefore, the most general definition of a type-1 encryption operator would look like:
\begin{equation*}
	U_{\enc} : \ket{\pk,\rnd,\mssg,y} \mapsto \ket{\pk,\rnd,\mssg,y\xor\enc_\pk(\mssg;\rnd)}\,.
\end{equation*}
We argue that this is indeed the most general and correct way to model the local 
computational power of a quantum adversary, even in the QS1 case.
However, for ease of exposition (and also because it would go beyond the traditional meaning of ciphertext indistinguishability), in the present work we do not consider superpositions of public keys, as we assume that the (classical) public key to be attacked is given to the adversary at the beginning of the security game.
Hence, we drop the register containing the public key and consider it a parameter of the unitary.
This leads us to the following definition.
\begin{definition}[Type-1 Encryption for PKE]\label{def:type1enc}
	Let $\Sigma=(\Kgen,\Enc,\Dec)$ be a PKE scheme and let $(\pk,\sk) \getsr \KGen$. The {\em type-1 encryption operator} for $\pk$ is the unitary defined by:
	\begin{equation*}
		U^{(1)}_{\Enc_\pk} : \ket{\rnd,\mssg,y} \mapsto \ket{\rnd,\mssg,y\xor\enc_\pk(\mssg;\rnd)}\,.
	\end{equation*}
\end{definition}
Usually the public key is clear from the context, so we will omit that dependency and just write $U^{(1)}_{\Enc}$.
As usual, when there is no ambiguity, we identify the corresponding superoperator acting on mixed states rather than pure states with the same symbol $U^{(1)}_\Enc : \states(\hilbert_\rndsp \otimes \hilbert_\msgsp \otimes \hilbert_\ctxsp) \to \states(\hilbert_\rndsp \otimes \hilbert_\msgsp \otimes \hilbert_\ctxsp)$.
By letting the randomness be an input, Definition~\ref{def:type1enc} allows to encrypt using a superposition of randomnesses, 
which is fine in the case of the adversary generating ciphertexts himself. 
Note also that the case of different randomnesses for each message in superposition can be realised by using a single classical randomness and a QS2-secure pseudorandom function~\cite{FOCS:Zhandry12}, as shown by Boneh and Zhandry~\cite{C:BonZha13}. 

The type-1 decryption operator is defined analogously to~Definition~\ref{def:type1enc}, but with an important difference:
the decryption algorithm does not take the randomness used for encryption as input.
\begin{definition}[Type-1 Decryption for PKE]\label{def:type1dec}
	Let $\Sigma=(\Kgen,\Enc,\Dec)$ be a PKE scheme and let $(\pk,\sk) \getsr \KGen$. The {\em type-1 decryption operator} for $\sk$ is the unitary defined by:
	\begin{equation*}
		U^{(1)}_{\Dec_\sk} : \ket{\ctxt,z} \mapsto \ket{\ctxt,z\xor\Dec_\sk(\ctxt)}\,.
	\end{equation*}
\end{definition}
As usual we denote it by $U^{(1)}_{\Dec}$, leaving the secret key understood, and when there is no ambiguity with the same symbol we denote the superoperator acting on mixed states also by $U^{(1)}_\Dec : \states(\hilbert_\ctxsp \otimes \hilbert_\msgsp) \to \states(\hilbert_\ctxsp \otimes \hilbert_\msgsp)$.

Notice the difference in type-1 encryption and decryption acting on different spaces: this is not surprising, as it is already known that the adjoint of a type-1 encryption operator is not, generally, a type-1 decryption operator.
Notice also how both operators are efficiently computable, because \Enc and \Dec are efficient algorithms.
The difference is that realising $U^{(1)}_{\Dec}$ requires knowledge of the secret key \sk, while for realising $U^{(1)}_\Enc$ it is sufficient to know the public key \pk.

\subsection{Type-2 Encryption for PKE}\label{sec:type2PKE}

When defining type-2 encryption for PKE schemes, we have to remember that defining these operators only makes sense for functions which are reversible.
If a PKE scheme is perfectly correct, then encryption is always reversible if seen as a function of the {\em plaintext},  
but not necessarily as a function of the {\em randomness}.
That is because it might be the case that for a given message different 
randomnesses lead to the same ciphertext.
In the context of security games, message and randomness have very different roles anyway, as one is generally chosen by the adversary, while the other is generally chosen by the challenger.

Ultimately, what we want is to define a type of unitary which generalises the case of arbitrary permutations from 
plaintext to ciphertext spaces (the same approach as considered in~\cite{C:GagHulSch16}).
In order to avoid the issue raised by randomness collisions, we will keep the auxiliary randomness register both in input and output of the circuit.
This ensures reversibility of the operator, because given a certain ciphertext and a certain randomness, there is only one possible plaintext which was mapped to that ciphertext (otherwise we would have a decryption failure, and for now we are only considering perfectly correct schemes).
So, if the sizes of the plaintext space and the ciphertext space coincide, i.e., there is no ciphertext expansion and 
thus $\dim(\hilbert_\msgsp) = \dim(\hilbert_\ctxsp)$, then we can define the corresponding type-2 encryption operator 
as:
\begin{equation*}
	U_\Enc^{(2)}: \ket{\rnd,\mssg} \mapsto \ket{\rnd,\Enc_\pk(\mssg;\rnd)}\,, 
\end{equation*}
where, as usual, the public key \pk is implicit in the definition of $U_\Enc^{(2)}$, i.e., it is a parameter of the unitary operator in question.

In the more general case of message expansion, i.e., $\dim(\hilbert_\msgsp) < \dim(\hilbert_\ctxsp)$, we use the same approach as in~\cite{C:GagHulSch16}: we introduce an auxiliary register in a complementary space $\hilbert_{\ctxsp - \msgsp}$\footnote{We denote by $\hilbert_{\ctxsp - \msgsp}$ a Hilbert space such that $\hilbert_\msgsp \otimes \hilbert_{\ctxsp - \msgsp}$ is isomorphic to $\hilbert_\ctxsp$. Notice that the opposite case, i.e., $\dim(\hilbert_\msgsp) > \dim(\hilbert_\ctxsp)$, cannot happen because it would lead to collisions on the ciphertexts and thus introduce decryption failures. Also notice that, as in~\cite{C:GagHulSch16}, the case of adversarially-controlled ancilla qubits is left as an open problem.} that ensures reversibility of the operation, and which is initialised to $\ket{0\ldots 0}$ during an honest execution to yield a correct encryption. So we consider a family of unitary superoperators of the form:
\begin{align*}
	U: \states(\hilbert_\rndsp \otimes \hilbert_\msgsp \otimes \hilbert_{\ctxsp - \msgsp}) &\to \states(\hilbert_\rndsp \otimes \hilbert_\ctxsp) , \text{ such that } \\
	U: \ketbra{\rnd, \mssg, y} &\mapsto \psi\,, 
\end{align*}
and we define a type-2 encryption operator to be any arbitrary, efficiently computable (purified) representative of the above family such that:
\begin{equation}\label{eqn:type2canon}
	U_\Enc^{(2)}: \ket{\rnd, \mssg, 0\ldots 0} \mapsto \ket{\rnd, \Enc_\pk(\mssg; \rnd)}.
\end{equation}
The choice of the particular representative is irrelevant in our exposition as long as it respects~\eqref{eqn:type2canon} above and it is efficiently computable.
However, as already discussed in~\cite{C:GagHulSch16}, it might be the case that realising this operator requires knowledge of the secret key, not only of the public key.
This finally leads to the following.
\begin{definition}[Type-2 Encryption for PKE]\label{def:type2enc}
	Let $\Sigma =(\kgen,\enc,\dec)$ be a perfectly correct PKE scheme, and let $(\pk,\sk) \getsr \KGen$. 
	A 
	{\em type-2 encryption operator} for $\Sigma$ is an efficiently computable unitary in the family defined by:
	\begin{equation*}
		U_{(\Enc,\pk,\sk)}^{(2)}: \ket{\rnd, \mssg, 0\ldots 0} \mapsto \ket{\rnd, \Enc_\pk(\mssg; \rnd)}\,.
	\end{equation*}
	It will be usually denoted by just $U_\Enc^{(2)}$ when there is no ambiguity.
\end{definition}
It is always possible to find and efficiently sample and implement 
at least one 
valid representative for $U_\Enc^{(2)}$ given the secret and public keys, by using a conversion circuit of type-1 encryption and decryption operators in a similar way as presented in~\cite{C:GagHulSch16}. 
We call this the {\em canonical} type-2 operator.
\begin{theorem}[Efficient Realisation of Type-2 Encryption]\label{thm:type2enceff}
	Let $\Sigma$ be a perfectly correct PKE scheme with $\Sigma = (\kgen,\enc,\dec)$, and let $(\pk,\sk) \getsr \KGen$. Then there exists an efficient procedure which takes \pk and \sk as input, and outputs a polynomial-size quantum circuit realising $U_\Enc^{(2)}$.
\end{theorem}
\begin{proof}
	The explicit circuit of the procedure is shown in Fig.~\ref{fig:type2enccanonical}.
	It uses type-1 encryption and decryption operators as underlying components, which are both efficient with knowledge of the respective keys.
	\begin{figure}[htbp]
		\centering
		\scalebox{\scaleFactorGraphics}{
			\begin{tikzpicture}
				
				\node (MIDDLE) [text width=1em,align=center] {};
				
				\node (Uenc)[text width=2em,align=center,fill=gray!25,minimum height=3.0cm,minimum width=2cm,draw] [left of=MIDDLE, node distance = 2.0cm] {\large $U_{\Enc}^{(1)}$};
				\node (Urec)[text width=2em,align=center,fill=gray!25,minimum height=2.0cm,minimum width=2cm,yshift =-0.5cm,draw] [right of=MIDDLE, node distance = 2.0cm] {\large $U_{\Dec}^{(1)}$};

				\node (IN1) [text width=1em,align=center] [left of=Uenc, yshift=1.0cm, node distance=3.5cm] {\large \ket{\rnd}};
				\node (IN2) [text width=1em,align=center] [left of=Uenc, yshift=0.0cm, node distance=3.5cm] {\large \ket{\mssg}};
				\node (IN3) [text width=1em,align=center] [left of=Uenc, yshift=-1.0cm, node distance=2.0cm] {\large \ket{0}};
				
				\draw[] ($(IN1.east) + (0.1,0.0)$) -- ($(Uenc.west) + (0.0,1.0)$);
				\draw[] ($(IN2.east) + (0.2,0.0)$) -- ($(Uenc.west) + (0.0,0.0)$);
				\draw[] ($(IN3.east) + (0.1,0.0)$) -- ($(Uenc.west) + (0.0,-1.0)$);
				
				\node (OUT1) [text width=1em,align=center] [right of=Urec, yshift=1.5cm, node distance=3.5cm] {\large \ket{\rnd}};
				\node (OUT2) [text width=1em,align=center] [right of=Urec, yshift=0.5cm, node distance=3.5cm] {\large \ket{\ctxt}};
				\node (OUT3) [text width=1em,align=center] [right of=Urec, yshift=-0.5cm, node distance=2.0cm] {\large \ket{0}};
				
				\draw[] ($(Urec.east) + (0.0,0.5)$) -- ($(OUT2.west) + (0.0,0.0)$);
				\draw[] ($(Urec.east) + (0.0,-0.5)$) -- ($(OUT3.west) + (0.0,0.0)$);
				
				\node (midOUT1) [text width=1em,align=center] [right of=Uenc, yshift=1.0cm, node distance=1.0cm] {};
				\node (midOUT2) [text width=1em,align=center] [right of=Uenc, yshift=0.0cm, node distance=1.0cm] {};
				\node (midOUT3) [text width=1em,align=center] [right of=Uenc, yshift=-1.0cm, node distance=1.0cm] {};
				\node (midIN1) [text width=1em,align=center] [left of=Urec, yshift=0.5cm, node distance=1.0cm] {};
				\node (midIN2) [text width=1em,align=center] [left of=Urec, yshift=-0.5cm, node distance=1.0cm] {};
				
				\draw[] ($(Uenc.east) + (0.0,1.0)$) -- ($(OUT1.west) + (0.0,0.0)$);
				
				\draw ($(midOUT2)$) to[out=0,in=-180] ($(midIN2)$);
				\draw ($(midOUT3)$) to[out=0,in=-180] ($(midIN1)$);
				
				\begin{scope}[on background layer]        
					\node (type2Box) [text width=3em,rounded corners=1ex,minimum height=3.5cm,minimum width=9.0cm,draw, densely dashed,line width = 0.3mm] at ($(MIDDLE)$) {};
				\end{scope}
				
			\end{tikzpicture}
		}
		\caption{Canonical type-2 encryption operator for perfectly correct PKE schemes.}
		\label{fig:type2enccanonical}
	\end{figure}
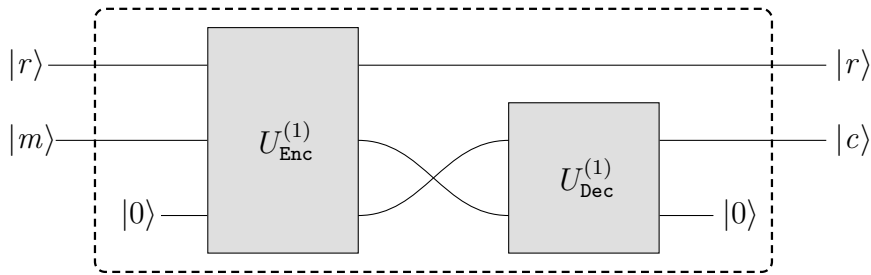
	\qed
\end{proof}
Notice that realising this canonical type-2 operator requires knowledge of the secret key, even if it is just an encryption operator, but that is fine because as previously mentioned type-2 operators usually require this additional knowledge. 
We have to make a distinction between the encryption {\em unitary} as defined above (a quantum gate modelling local computation of encryption by a party with knowledge of the relevant keys) and the encryption {\em oracle} (modelling the interaction of the adversary with such party, usually the challenger). 
By letting the randomness be an input, Definition~\ref{def:type2enc} allows to encrypt using a superposition of randomnesses, 
which is fine in the case of a party generating ciphertexts himself. 
In our security notion, however, the (honest) challenger will always produce ciphertexts using a (secret) classical randomness not controlled by the adversary \Eve.
In the security game, the challenger cannot send the randomness register back to \Eve, because knowledge of the randomness used would trivially break security, even in a classical scenario. 
But at the same time if the challenger withholds the randomness register, from \Eve's perspective this would be equivalent to tracing it out, and if the type-2 encryption operator introduces entanglement between ciphertext and randomness output registers, then tracing out the randomness would disturb the ciphertext state.

Luckily, a simple observation solves this dilemma:
as we have already discussed, in our oracle case the randomness is chosen by the (honest) challenger during the challenge query, so we can safely model it as classical.\footnote{Even if considering challengers that use superpositions of randomnesses, we show in Appendix~\ref{app:randsuperposition} that the difference is irrelevant, and that we can always restrict ourselves to the case of a classical randomness register.}
Looking at Definition~\ref{def:type2enc}, this means that the output state is always separable as $\ketbra{\rnd} \otimes \psi$.
Therefore, in our oracle definition the randomness register can be discarded after applying the type-2 encryption without disturbing the ciphertext state. 
This leads to the following.
\begin{definition}[Type-2 Encryption Oracle]\label{def:type2encOracle}
	Let $\Sigma=(\Kgen,\Enc,\Dec)$ be a PKE scheme and let $(\pk,\sk) \getsr \KGen$. The {\em type-2 encryption oracle} $O^{(2)}_{\Enc}$ for $\pk$ is defined by the following procedure:
	
	\begin{pchstack}
		\begin{pcvstack}\pcvspace
			\procedure[linenumbering]{\pcoracle $O^{(2)}_{\Enc}(\phi)$ on input $\phi \in \states(\hilbert_\msgsp)$}{%
				\x \rnd \fromunif \rndsp \\
				\x \ketbra{\rnd} \tensor \psi \gets U^{(2)}_\Enc \left( \ketbra{\rnd} \tensor \phi \tensor \ketbra{0 \ldots 0} \right)  \\
				\x \text{trace out } \ketbra{\rnd} \\
				\x \pcreturn \psi
			}
		\end{pcvstack}
	\end{pchstack}
\end{definition}
%

\subsection{Recoverable PKE Schemes}\label{sec:recoverablePKE}

Now we introduce a special case of PKE schemes where it is possible to decrypt a ciphertext {\em without} knowledge of the secret key, but having access to the randomness used for the encryption instead.
These schemes might not be perfectly correct, so the decryption procedure might fail on some ciphertext, yet still the recovery procedure will `decrypt' correctly if the right randomness is provided. 
We will see in Section~\ref{sec:SecurityAnalysis} that many PKE schemes are actually of this type.
\begin{definition}[Recoverable PKE Scheme]\label{def:recoverablePKEscheme}
	Let $\Sigma = (\kgen,\enc,\dec)$ be a (not necessarily perfectly correct) PKE scheme.
	We call $\Sigma$ a \emph{recoverable PKE scheme} if there exists an efficient algorithm $\Rec:\pksp \times \rndsp \times \ctxsp \to \msgsp$ such that, for any $\pk \in \pksp, \rnd \in \rndsp, \mssg \in \mssgsp$, it holds that
	\begin{equation*}
		\Rec(\pk,\rnd,\enc_\pk(\mssg;\rnd)) = \mssg\,.
	\end{equation*}
\end{definition}
Notice how the recovery procedure will always allow to avoid decryption failures even for schemes which are not perfectly correct.
We will sometimes write a recoverable scheme $\Sigma = (\KGen,\Enc,\Dec)$ with recovery algorithm \Rec directly as $\Sigma = (\KGen,\Enc,\Dec,\Rec)$.
Given \pk, it is of course possible to define a type-1 operator for \Rec in the canonical way.
\begin{definition}[Type-1 Recovery for PKE]\label{def:type1rec}
	Let $\Sigma=(\Kgen,\Enc,\Dec,\Rec)$ be a recoverable PKE scheme, and let $(\pk,\sk) \getsr \KGen$.
	The \emph{type-1 recovery operator} for $\pk$ is the unitary defined by:
	\begin{equation*}
		U^{(1)}_{\Rec_\pk} : \ket{\rnd,\ctxt,z} \mapsto \ket{\rnd,\ctxt, z \xor \Rec_\pk(\rnd,\ctxt)}\,.
	\end{equation*}
\end{definition}
As usual we will denote this operator by $U^{(1)}_{\Rec}$ when there is no ambiguity in the choice of \pk, and  with the same symbol we denote the superoperator acting on mixed states, i.e., $U^{(1)}_\Rec : \states(\hilbert_\rndsp \otimes \hilbert_\ctxsp \otimes \hilbert_\msgsp) \to (\hilbert_\rndsp \otimes \hilbert_\ctxsp \otimes \hilbert_\msgsp)$.

Now, the crucial observation is the following: for recoverable PKE schemes, the canonical type-2 encryption operator can be efficiently implemented using only the public key.
\begin{theorem}[Type-2 Encryption Operator for Recoverable Schemes]\label{thm:recoverabletype2}
	Let \mbox{$\Sigma = (\Kgen, \Enc, \Dec, \Rec)$} be a recoverable PKE, and let $(\pk,\sk) \getsr \KGen$.
	Then there exists an efficient procedure which only takes \pk as input, and outputs a polynomial-size quantum circuit realising the canonical operator $U_\Enc^{(2)}$.
\end{theorem}
\begin{proof}
	The explicit circuit of the procedure is shown in Fig.~\ref{fig:type2encrecovery}.
	It uses type-1 encryption and recovery operators as underlying components, which are both efficient with knowledge of the public key only. Realisation of both these components is independent of the fact whether the scheme has full correctness or not, as the decryption algorithm itself is never used.
	\begin{figure}[htbp]
		\centering
		\scalebox{\scaleFactorGraphics}{
			\begin{tikzpicture}
				
				\node (MIDDLE) [text width=1em,align=center] {};
				
				\node (Uenc)[text width=2em,align=center,fill=gray!25,minimum height=3.0cm,minimum width=2cm,draw] [left of=MIDDLE, node distance = 2.0cm] {\large $U_{\Enc}^{(1)}$};
				\node (Urec)[text width=2em,align=center,fill=gray!25,minimum height=3.0cm,minimum width=2cm,draw] [right of=MIDDLE, node distance = 2.0cm] {\large $U_{\Rec}^{(1)}$};

				\node (IN1) [text width=1em,align=center] [left of=Uenc, yshift=1.0cm, node distance=3.5cm] {\large \ket{\rnd}};
				\node (IN2) [text width=1em,align=center] [left of=Uenc, yshift=0.0cm, node distance=3.5cm] {\large \ket{\mssg}};
				\node (IN3) [text width=1em,align=center] [left of=Uenc, yshift=-1.0cm, node distance=2.0cm] {\large \ket{0}};
				
				\draw[] ($(IN1.east) + (0.1,0.0)$) -- ($(Uenc.west) + (0.0,1.0)$);
				\draw[] ($(IN2.east) + (0.2,0.0)$) -- ($(Uenc.west) + (0.0,0.0)$);
				\draw[] ($(IN3.east) + (0.1,0.0)$) -- ($(Uenc.west) + (0.0,-1.0)$);
				
				\node (OUT1) [text width=1em,align=center] [right of=Urec, yshift=1.0cm, node distance=3.5cm] {\large \ket{\rnd}};
				\node (OUT2) [text width=1em,align=center] [right of=Urec, yshift=0.0cm, node distance=3.5cm] {\large \ket{\ctxt}};
				\node (OUT3) [text width=1em,align=center] [right of=Urec, yshift=-1.0cm, node distance=2.0cm] {\large \ket{0}};
				
				\draw[] ($(Urec.east) + (0.0,1.0)$) -- ($(OUT1.west) + (0.0,0.0)$);
				\draw[] ($(Urec.east) + (0.0,0.0)$) -- ($(OUT2.west) + (0.0,0.0)$);
				\draw[] ($(Urec.east) + (0.0,-1.0)$) -- ($(OUT3.west) + (0.0,0.0)$);
				
				\node (midOUT1) [text width=1em,align=center] [right of=Uenc, yshift=1.0cm, node distance=1.0cm] {};
				\node (midOUT2) [text width=1em,align=center] [right of=Uenc, yshift=0.0cm, node distance=1.0cm] {};
				\node (midOUT3) [text width=1em,align=center] [right of=Uenc, yshift=-1.0cm, node distance=1.0cm] {};
				\node (midIN1) [text width=1em,align=center] [left of=Urec, yshift=1.0cm, node distance=1.0cm] {};
				\node (midIN2) [text width=1em,align=center] [left of=Urec, yshift=0.0cm, node distance=1.0cm] {};
				\node (midIN3) [text width=1em,align=center] [left of=Urec, yshift=-1.0cm, node distance=1.0cm] {};
				
				\draw[] ($(Uenc.east) + (0.0,1.0)$) -- ($(Urec.west) + (0.0,1.0)$);
				
				\draw ($(midOUT2)$) to[out=0,in=-180] ($(midIN3)$);
				\draw ($(midOUT3)$) to[out=0,in=-180] ($(midIN2)$);
				
				\begin{scope}[on background layer]        
					\node (type2Box) [text width=3em,rounded corners=1ex,minimum height=3.5cm,minimum width=9.0cm,draw, densely dashed,line width = 0.3mm] at ($(MIDDLE)$) {};
				\end{scope}
				
			\end{tikzpicture}
		}
		\caption{Canonical type-2 encryption operator for recoverable PKE schemes.}
		\label{fig:type2encrecovery}
	\end{figure}
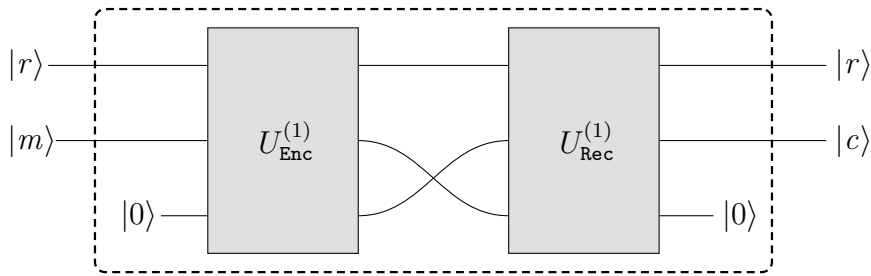
	\qed
\end{proof}
In particular, for recoverable PKE schemes the type-2 encryption operator can be realised locally by a quantum adversary (or a reduction), without need of additional oracle access.
This, together with the fact that most real-world PKE schemes are recoverable (as we will see in Section~\ref{sec:SecurityAnalysis}) shows that type-2 encryption operators are very natural, and unlike in the symmetric key case considered in~\cite{C:GagHulSch16} they also appear implicitly in QS1 security notions for such schemes.

\subsection{The \qINDqCPA Security Notion}\label{sec:qINDqCPA}

We are now ready to define the notion of {\em quantum ciphertext indistinguishability under quantum chosen plaintext 
	attack} (\qINDqCPA) for PKE schemes which admit an efficient construction of the canonical type-2 encryption operator $U_\Enc^{(2)}$. This includes in particular perfectly correct schemes and recoverable schemes.\footnote{As we will see, these cover all the interesting cases in practice, although there might be other classes of schemes which allow an efficient construction of $U_\Enc^{(2)}$; 
	we address the general case in Section~\ref{sec:classif}.}
We follow the approach in~\cite{C:GagHulSch16} and we define a game where a polynomially bounded quantum adversary plays against an external challenger.
We have to define the challenge phase and the learning (quantum \CPA) phases (pre- and post-challenge), using the 
theory of type-2 operators we have devised so far.

For the challenge query it is pretty straightforward: as in the original \qIND security definition for symmetric key encryption, we assume that the challenger \chal generates a keypair and sends the public key \pk to the adversary \Eve.
Then \Eve sends two plaintext quantum states (possibly mixed) $\phi_0,\phi_1$ to \chal, who will 
flip a random bit $b \getsr \bit$, discard (trace out) $\phi_{1-b}$, and encrypt the other message with the type-2 encryption oracle $\psi \from O^{(2)}_\Enc(\phi_b)$.
Finally, $\psi$ is sent back to \Eve, who will have to guess $b$ in order to win the game.

Justifying the use of a type-2 encryption during the challenge phase requires arguments different from the symmetric key case.
In the classical \INDCPA game for PKE, the challenger does not even need to know the secret key, as it is not needed for encryption, and we saw already that the secret key is sometimes necessary to implement the canonical type-2 encryption operator.
However, in the QS2 case the challenger can produce ciphertext-encoding quantum states with very different structure depending on whether he knows the secret key or not, thereby leading to different attack models. Type-2 encryption operators in particular are more general in this respect, and allow us to aim for a stronger security notion.
Moreover we also saw how certain schemes, like the recoverable ones, allow to build the type-2 operator using only the public key.
Thus it makes sense for a QS2 security notion to include the use of type-2 operators during the challenge phase.

The other question we have to address, which was left unspecified in~\cite{C:GagHulSch16}, is about the learning (\qCPA) phase.
Shall the adversary be able to perform only type-1 encryption operations, or type-2 as well?
In the QS1 case the answer is obvious: it depends on the scheme, e.g., for recoverable schemes both \mbox{type-1} and \mbox{type-2} operations are allowed, but in the general case only type-1 operations should.
Instead, in the QS2 case that we are considering, the answer is less straightforward.
For recoverable schemes again there is no difference, as the adversary can implement both types of operators locally.
But for general schemes there might be a difference, and there might exist non-recoverable PKE schemes which become insecure when giving oracle access to a type-2 encryption operator during the learning phase.\footnote{For example, one could combine a suitable separating SKE scheme with the canonical hybrid construction
	(cf. Section~\ref{sec:SecurityAnalysisHybrid}),
	so that the separation property is `inherited' by the resulting PKE scheme. We are not aware of an explicit example of such SKE scheme and we leave this as an open problem. We stress that such a counterexample is not found in~\cite{eprint:CETU20}, as the authors there ``excluded [...] notations that [...] combine quantum learning queries with quantum challenge queries of different query models.''}

In our definition of \qINDqCPA we opt for giving to the adversary as much power as possible, hence explicitly giving access to a type-2 encryption oracle when dealing with non-recoverable schemes, both in the learning and challenge phases.
The reason for this choice is twofold.
First, this allows us to aim for potentially stronger security notions.
Second, remember that, classically, \CPA attacks model not only the case where the adversary can compute ciphertexts himself (as in the case of PKE), but also scenarios where the adversary can ``trick'' an honest encryptor in providing certain ciphertexts (as in the case of \INDCPA security for symmetric key encryption).
In the quantum PKE setting, there is a difference whether these ciphertexts are computed locally by the adversary or obtained by the challenger through ``trickery'' (including scenarios already considered in~\cite{C:GagHulSch16}, such as quantum side-channel attacks, quantum obfuscation, etc.), because the challenger has knowledge of the secret key, and is therefore capable of generating type-2 ciphertexts even if the scheme is non-recoverable.
So, giving the adversary access to the type-2 encryption oracle seems to be the ``safe'' choice.

These considerations finally lead to the following.
\begin{exper}\label{exp:qINDqCPA}
	The {\em \qINDqCPA experiment} $\gamefont{qIND\mbox{-}qCPA}(\Sigma, \Eve, \secpar)$ for a PKE scheme $\Sigma=(\KGen,\Enc,\Dec)$ is defined as follows:
	\begin{algorithmic}[1]
		\State \chal runs $(\pk,\sk) \from \KGen$ and implements $O^{(2)}_\Enc$
		\State $\Eve^{O^{(2)}_\Enc}(\pk) \to (\phi_0, \phi_1, \sigma_{\varfont{state}})$
		\State \chal receives $\phi_0, \phi_1$ and does the following: \begin{itemize}
			\item[-] flips $b \fromunif \bit$
			\item[-] traces out $\phi_{1-b}$
			\item[-] calls $\psi \from O^{(2)}_\Enc(\phi_b)$
			\item[-] sends $\psi$ to \Eve
		\end{itemize}
		\State $\Eve^{O^{(2)}_\Enc}(\sigma_{\varfont{state}},\psi) \to b' \in \bit$
		\State \textbf{if} $b = b'$ \textbf{then return} \win; \textbf{else return} \rej.
	\end{algorithmic}
\end{exper}
Security is defined as negligible advantage over guessing.
\begin{definition}[\qINDqCPA Security]\label{def:qINDqCPA}
	A public key encryption scheme $\Sigma$ has {\em quantum ciphertext indistinguishability under quantum chosen plaintext attack}, or is \qINDqCPA-secure, iff for any QPT adversary \Eve it holds:
	\begin{equation*}
		\left| \Pr\left[ \gamefont{qIND\mbox{-}qCPA}(\Sigma, \Eve, \secpar) \to \win \right]- \frac{1}{2}\right| \leq \negl(\secpar)\,.
	\end{equation*}
\end{definition}
It is easy to show that the above notion is at least as strong as the QS1 notion of \INDqCPA for PKE introduced in~\cite{C:BonZha13}.
Before we show this, let us first recall some game-based notation~\cite{EC:BelRog06,eprint:Shoup04}.
Let \game be a game (or experiment) instantiated with a cryptographic scheme $\Sigma$ and \Eve be an adversary.
We write \Adv[\Sigma]{\game}{\Eve} to denote the advantage of \Eve in game \game instantiated with $\Sigma$, e.g., $\Adv[\Sigma]{\qINDqCPA}{\Eve}$ for the \qINDqCPA advantage against $\Sigma$.
If the scheme is clear from the context, we simply write $\Adv{\game}{\Eve}$.
For games \game[1] and \game[2], we write $\gameAdv{\game[1]^{\Eve}}{\game[2]^{\Eve}}$ for the advantage of adversary \Eve in distinguishing the games.
\begin{theorem}[$\qINDqCPA \Rightarrow \INDqCPA$]\label{thm:qINDqCPA_implies_INDqCPA}
	Let \mbox{$\Sigma = (\Kgen, \Enc, \Dec)$} be a PKE scheme.
	For any adversary \Eve, it holds that
	\begin{align*}
		\Adv{\INDqCPA}{\Eve} \le \Adv{\qINDqCPA}{\Eve}\,.
	\end{align*}
\end{theorem}
\begin{proof}
	We show that any adversary \adver wins the \qINDqCPA game with at least the same probability of winning the \INDqCPA game; the latter (Experiment~\ref{exp:PKEINDqCPA}) is described in Appendix~\ref{app:additionalBackground}. The differences with Experiment~\ref{exp:qINDqCPA} are:
	\begin{enumerate}
		\item In the \INDqCPA game \adver does not get oracle access to $O^{(2)}_\Enc$. Hence, when switching to \qINDqCPA, the winning probability cannot decrease, because the power of the adversary is augmented by the type-2 oracle.
		\item In the \INDqCPA game \adver is restricted to classical challenge messages $\mssg[0], \mssg[1]$. When switching to \qINDqCPA, the adversary will simply use quantum states $\ket{\mssg[0]}, \ket{\mssg[1]}$ as challenge plaintexts instead, and will measure the quantum ciphertext received by the challenger.
	\end{enumerate}
	Notice in fact that, since the randomness \rnd in the \qINDqCPA challenge query is classical, the type-2 operator $U^{(2)}_\Enc$ will produce a ciphertext state which is just a classical ciphertext encoded as a basis state $\ket{c = \Enc_\pk(\mssg;\rnd)}$.
	In other words, quantum plaintexts are {\em more generic} than classical plaintexts (or, to put it differently, classical plaintexts are a very special case of quantum plaintexts), and hence again the power of the adversary is not diminished when switching to the \qINDqCPA game.
	\qed
\end{proof}

\subsection{The \CCA Case}\label{sec:qINDqCCA}

We leave the case of extending our exposition to the quantum chosen ciphertext attack case (with the relevant notions of \qINDqCCAone and \qINDqCCAtwo) as future work, but we want anyway to sketch here the general strategy.

The first task is to formalise a \mbox{type-2} operator for decryption.
Unlike in the symmetric key setting considered in~\cite{C:GagHulSch16}, this is not necessarily going to be the adjoint of the \mbox{type-2} encryption operator, and in particular it might not require a randomness register as input; this has to be expected given that there is already an asymmetry in the definition of \mbox{type-1} encryption and decryption operators in the public key setting. 
Then, in the \qINDqCCAone case, we just extend the \qINDqCPA experiment by also providing the adversary with oracle access to the \mbox{type-1} and \mbox{type-2} decryption operators.

Extending the framework to the \qINDqCCAtwo case is not straightforward, mainly due to no-cloning and the destructive nature of quantum measurement. 
In fact, this case was left as an open problem already in~\cite{C:GagHulSch16} for the symmetric key setting. 
Fortunately, the technique presented in~\cite{EC:AlaGagMaj18} shows how to overcome this difficulty, by using a real-VS-ideal approach which makes it possible to differentiate the behaviour of the adversary when replaying the challenge ciphertext to the decryption oracle, hence effectively detecting a challenge replay attack. 
The approach in~\cite{EC:AlaGagMaj18} (and its extension to the public key case presented in~\cite{DBLP:journals/iacr/AlagicGM18}) is given in the context of {\em quantum encryption schemes} (a scenario which falls under the QS3 domain in~\cite{DBLP:phd/dnb/Gagliardoni17}), but it is easy to generalise to the QS2 notions we are considering here.

\section{Security Analysis for Real-World PKE Schemes} \label{sec:SecurityAnalysis}

We analyse the \qINDqCPA security of several real-world public key encryption schemes.
We start with the canonical LWE-based PKE scheme in Section~\ref{sec:SecurityAnalysisLatticeBased}, followed by the code-based PKE scheme ROLLO-II in Section~\ref{sec:SecurityAnalysisCodeBased}.
The hybrid encryption scheme is analysed in Section~\ref{sec:SecurityAnalysisHybrid} while Section~\ref{sec:SecurityAnalysisDiscussion} concludes with a discussion of these results.

\subsection{Results for LWE-Based PKE} \label{sec:SecurityAnalysisLatticeBased}

In this section we analyse the canonical LWE-based public key encryption scheme due to Regev~\cite{STOC:Regev05} with respect to our \qINDqCPA security notion.

\subsubsection{LWE-Based Public Key Encryption Schemes.} \label{sec:LWEscheme}

The canonical LWE-based encryption scheme has been proposed by Regev~\cite{STOC:Regev05}.
It underpins most lattice-based PKE schemes such as Kyber~\cite{DBLP:conf/eurosp/BosDKLLSSSS18}, LIMA~\cite{ESORICS:AOPPS17}, the LP scheme~\cite{RSA:LinPei11}, and the schemes underlying NewHope~\cite{NISTPQC-R2:NewHope19} and LAC~\cite{NISTPQC-R2:LAC19}.
The pseudocode (that we give for simplicity in a generic form, i.e., not specifying concrete domains and distributions for the parameters) is given in Fig.~\ref{fig:canonicalLWEbasedPKEscheme}.
Its security is based on the computational hardness of the \emph{Learning With Errors (LWE)} lattice problem.
The canonical LWE-based scheme works on $n$-dimensional vectors of elements of $\mathbb{Z}_q$ for $q \geq 2$. 
The functions $\mathsf{Encode}$ and $\mathsf{Decode}$ are used for encoding and decoding bit strings 
to and from elements of $(\mathbb{Z}_q)^n$. 
The $\mathsf{Decode}$ function has a certain error tolerance which, upon being exceeded, results in a decryption failure.

The $\mathsf{Encode}$ function maps the bits of the message into the high-order bit representation of group elements, which are then represented as a vector. 
For our purpose, it is not important here to have a precise definition of this encoding function, nor to have a detailed discussion on the the sampling distribution of the LWE vectors -- which is generally crucial for proving the security of the scheme. We leave these details to an appropriate reference, for example~\cite{NISTPQC-R2:FrodoKEM19}.
In this section, we will only consider the following, simplified characterisation of the encoding function.
\begin{lemma}[Canonical LWE-Based Message Encoding Representation]\label{lem:encodeRepr}
	Let $\mathsf{Encode}: \mssg \mapsto \poly{v}$ as from Fig.~\ref{fig:canonicalLWEbasedPKEscheme}, and let $\mathsf{Bit}(\poly{v})$ the canonical bit string representation in use for a vector element $\poly{v}$
	over a finite group $\mathbb{Z}_q$. 
	Then there exists a public efficient invertible permutation $\pi$ 
	and an integer $\tau \geq 0$ 
	such that
	\begin{align*}
		\mathsf{Bit}(\mathsf{Encode}(\mssg)) = \pi \left( \mssg \| \overbrace{0 \ldots 0}^{\tau } \right) \,.
	\end{align*}
	In particular, for $q=2$, it holds $\tau = 0$ and $\mathsf{Bit}(\mathsf{Encode}(\mssg)) = \pi(\mssg)$.
\end{lemma}
Notice that in practice the parameter $\tau$ denotes the expansion factor between $\mssg$ and $\ctxt[1]$ in Fig.~\ref{fig:canonicalLWEbasedPKEscheme} and is upper bounded by $n \cdot \ceil{\log_2(q)}$ minus the bit size of the message.
The larger $\tau$, the less efficient the scheme is in terms of ciphertext size but the lower the decryption failure rate.

\begin{figure}[htbp]
	\centering
		\begin{pchstack}
			\fbox{
				\procedure[width=9em]{$\kgen(\lambda;\rnd)$}{%
					\x \poly{a}, \poly{s},\poly{e} \gets \rnd \\
					\x \poly{b} \gets \poly{a}\poly{s} + \poly{e} \\
					\x \pk \gets (\poly{a},\poly{b}) \\
					\x \sk \gets \poly{s} \\
					\x \pcreturn (\sk,\pk)
				}
				\pchspace
				\procedure[width=11em]{$\enc(\pk,\mssg;\rnd)$}{%
					\x \pcparse \pk~\pcas (\poly{a},\poly{b}) \\
					\x \poly[1]{e},\poly[2]{e},\poly{d} \gets \rnd \\
					\x \ctxt[1] \gets \poly{b}\poly{d} + \poly[1]{e} + \mathsf{Encode}(\mssg) \\
					\x \ctxt[2] \gets \poly{a}\poly{d} + \poly[2]{e} \\
					\x \pcreturn \ctxt \gets (\ctxt[1],\ctxt[2])
				}		
			}
		\end{pchstack}
	\caption{Pseudocode of the canonical LWE-based public key encryption scheme $\Sigma = (\kgen,\enc,\dec)$. For the randomness \rnd used by \kgen and \enc, let $x \gets \rnd$ denote that $x$ is deterministically derived from \rnd. The decryption algorithm is omitted, as it is irrelevant for this work.}
	\label{fig:canonicalLWEbasedPKEscheme}
\end{figure}

Recall that \qINDqCPA security can only be defined for schemes which admit an efficient realisation of a type-2 encryption operator.
Showing this for the canonical LWE scheme is hence our first goal.
\begin{lemma}\label{lem:canonicalLWEbasedPKErecoverable}
	The canonical LWE-based PKE scheme $\Sigma = (\kgen,\enc,\dec)$, shown in Fig.~\ref{fig:canonicalLWEbasedPKEscheme}, is recoverable as from Definition~\ref{def:recoverablePKEscheme}.
\end{lemma}
\begin{proof}
	To prove the statement, we have to specify the algorithm \Rec that is introduced in Definition~\ref{def:recoverablePKEscheme}.
	Its input is a public key $\pk = (\poly{a},\poly{b})$, a randomness \rnd, and a ciphertext $\ctxt = (\ctxt[1],\ctxt[2])$ such that $\ctxt$ corresponds to the encryption of a message \mssg, using the public key \pk and randomness \rnd.
	The algorithm \Rec proceeds as follows.
	Given the randomness \rnd, it obtains the same values \poly[1]{e}, \poly[2]{e}, and \poly{d} that have been derived from \rnd during encryption and outputs
	\begin{align*}
		\mathsf{Decode}(\ctxt[1] - \poly{b}\poly{d} - \poly[1]{e})
		&= \mathsf{Decode}(\poly{b}\poly{d} + \poly[1]{e} + \mathsf{Encode}(\mssg) - \poly{b}\poly{d} - \poly[1]{e}) \\
		&= \mathsf{Decode}(\mathsf{Encode}(\mssg))
		= \mssg\,.
	\end{align*}
	This concludes the proof.
	\qed
\end{proof}

\subsubsection{QS2 Attack For LWE-Based PKE Schemes.}\label{sec:LWEattack}

Here we give an attack against the canonical LWE-based scheme for the case $q=2$. We leave the case of arbitrary $q$ as an open problem, albeit we conjecture that the distinguishing attack can be adapted to the general case.
\begin{theorem}[Attack Against Canonical LWE Scheme for $q=2$]\label{thm:LWEattack}
	Let $\Sigma$ be the canonical LWE-based PKE scheme shown in Fig.~\ref{fig:canonicalLWEbasedPKEscheme}
	defined over $\mathbb{Z}_q$ with $q=2$.
	Then there exists an efficient distinguishing adversary \Eve that
	wins the experiment $\gamefont{qIND\mbox{-}qCPA}(\Sigma,\Eve)$ with probability $1$.
\end{theorem}
\begin{proof}
	First of all recall that, because $q=2$, group elements are just represented as bits. This means that $\tau = 0$ (there is no padding nor message expansion) and $\mathsf{Bit}(\mathsf{Encode}(\mssg)) = \pi(\mssg)$. Moreover, addition is performed by XORing elements bitwise.
	The distinguishing adversary \Eve performs a single quantum challenge query using the Hadamard basis state $H\ket{0 \ldots 0}$ 
	(an uniform superposition of all messages)
	as a first quantum plaintext $\phi_0$, and 
	the state $H\ket{1 \ldots 1}$
	as a second quantum plaintext.
	$$
	\phi_0 = H\ket{0 \ldots 0} = \sum_{\mssg} \frac{1}{\sqrt{2^{|\mssg|}}} \ket{\mssg} , \ \ 
	\phi_1 = H\ket{1 \ldots 1} = \sum_{\mssg} \frac{1}{\sqrt{2^{|\mssg|}}} (-1)^{\mathsf{parity}(m)} \ket{\mssg} .
	$$
	Upon receiving back the ciphertext $\psi$, \Eve does the following:
	\begin{enumerate}
		\item traces out the second part of the ciphertext, corresponding to $\ctxt_2$ in Fig.~\ref{fig:canonicalLWEbasedPKEscheme};
		\item applies to the resulting state (a type-2 operator of) $\pi^{-1}$ (cf.~Lemma~\ref{lem:encodeRepr});
		\item measures the resulting state in the Hadamard basis;
		\item if the outcome is $0 \ldots 0$ then output $0$, otherwise output $1$.
	\end{enumerate}
	We now analyse the attack.
	If the challenge bit is $0$, then the state $\phi_0 = H\ket{0 \ldots 0}$ is encrypted.
	The resulting $\psi$ is:
	\begin{align*}
		O^{(2)}_\Enc \left( \sum_{\mssg} \frac{1}{\sqrt{2^{|\mssg|}}} \ket{\mssg} \right) 
		= \sum_{\mssg} \frac{1}{\sqrt{2^{|\mssg|}}} \ket{\Enc(\mssg)} 
		= \sum_{\mssg} \frac{1}{\sqrt{2^{|\mssg|}}} \ket{\ctxt_1^{(\mssg)}} \otimes \ket{\ctxt_2} 
		\,;
	\end{align*}
	where $\ctxt_1^{(\mssg)}$ is the $\ctxt_1$ part of the ciphertext (cf.~Fig.~\ref{fig:canonicalLWEbasedPKEscheme}) related to superposition element $\mssg$, while the second part $\ctxt_2$ is independent of the underlying plaintext.
	The two corresponding registers are hence unentangled, and after tracing out the second, \Eve has the state
	\begin{align*}
		\sum_{\mssg} \frac{1}{\sqrt{2^{|\mssg|}}} \ket{\ctxt_1^{(\mssg)}} = 
		\sum_{\mssg} \frac{1}{\sqrt{2^{|\mssg|}}} \ket{\poly{u} + \mathsf{Encode}(\mssg)} = 
		\sum_{\mssg} \frac{1}{\sqrt{2^{|\mssg|}}} \ket{\poly{u} + \pi(\mssg)}
		\,;
	\end{align*}
	for an unknown element $\poly{u}$. Looking at the bit representation of $\poly{u}$ and writing $\pi^{-1}(u) = w$, from Lemma~\ref{lem:encodeRepr} the state above can be written as
	\begin{align*}
		\sum_{\mssg} \frac{1}{\sqrt{2^{|\mssg|}}} \ket{u \oplus \pi(\mssg)}
		=& \sum_{\mssg} \frac{1}{\sqrt{2^{|\mssg|}}} \ket{ \pi(w \oplus \mssg)} 
		\,.
	\end{align*}
	Undoing $\pi$ finally yields:
	\begin{align*}
		\sum_{\mssg} \frac{1}{\sqrt{2^{|\mssg|}}} \ket{ w \oplus \mssg} = H\ket{0 \ldots 0}
		\,;
	\end{align*}
	which will give outcome $0 \ldots 0$ on a Hadamard measurement with probability $1$.
	
	On the other hand, if the challenge bit is $1$, then the state $\phi_1 = \ket{1 \ldots 1}$ is encrypted.
	A similar computation as before shows that the outcome of the encryption is another Hadamard state orthogonal to $H\ket{0 \ldots 0}$, so the outcome of \Eve's final measurement will be different from $0 \ldots 0$ with probability $1$.
	\qed
\end{proof}

\subsection{Results for Code-based PKE} \label{sec:SecurityAnalysisCodeBased}

In this section we analyse the code-based PKE scheme ROLLO-II~\cite{NISTPQC-R2:ROLLO19} with respect to our \qINDqCPA security notion.
It turns out that, due to the one-time pad encryption, ROLLO-II is not \qINDqCPA-secure.

\subsubsection{Code-Based Public Key Encryption ROLLO-II.}

The encryption scheme ROLLO-II~\cite{NISTPQC-R2:ROLLO19} is a code-based public key encryption scheme based on rank metric codes.
The scheme in a generic, simplified form is displayed in Fig.~\ref{fig:ROLLOII}, where \randomOracle is a random oracle, $\mathsf{Supp}$ describes the support of vectors, and \poly{P} is a polynomial from the underlying code problem.

\begin{figure}[htbp]
	\centering
		\begin{pchstack}
			\fbox{
				\procedure[width=9em]{$\kgen(\lambda;\rnd)$}{%
					\x \vec{x},\vec{y} \gets \rnd \\
					\x \vec{h} \gets \vec{x}^{-1}\vec{y} \bmod{\poly{P}} \\
					\x \sk \gets (\vec{x},\vec{y}) \\
					\x \pk \gets \vec{h} \\
					\x \pcreturn (\pk,\sk)
				}
				\pchspace
				\procedure[width=11em]{$\enc(\pk,\mssg;\rnd)$}{%
					\x \vec[1]{e},\vec[1]{e} \gets \rnd \\
					\x \mathsf{E} \gets \mathsf{Supp}(\vec[1]{e},\vec[2]{e}) \\
					\x \ctxt[1] \gets \mssg \xor \randomOracle(\mathsf{E}) \\
					\x \ctxt[2] \gets \vec[1]{e} + \vec[2]{e}\vec{h} \bmod{\poly{P}} \\
					\x \pcreturn \ctxt \gets (\ctxt[1],\ctxt[2])
				}	
			}
		\end{pchstack}
	\caption{Pseudocode of the code-based public key encryption scheme ROLLO-II. For the randomness \rnd used by \kgen and \enc, let $x \gets \rnd$ denote that $x$ is deterministically derived from \rnd. The decryption algorithm is omitted since it is irrelevant for our work.}
	\label{fig:ROLLOII}
\end{figure}

We first show that ROLLO-II is recoverable, and hence admits a \qINDqCPA security definition.
\begin{lemma}\label{lem:ROLLOIIrecoverable}
	The code-based PKE scheme ROLLO-II, shown in Fig.~\ref{fig:ROLLOII}, is recoverable as from Definition~\ref{def:recoverablePKEscheme}.
\end{lemma}
\begin{proof}
	To prove the statement, we have to specify the algorithm \Rec that is introduced in Definition~\ref{def:recoverablePKEscheme}.
	Its input is a public key $\pk = \vec{h}$, a randomness \rnd, and a ciphertext $\ctxt = (\ctxt[1],\ctxt[2])$, such that $\ctxt$ corresponds to the encryption of a message \mssg, using the public key \pk and randomness \rnd.
	The algorithm \Rec proceeds as follows.
	Given the randomness \rnd, it obtains the same values \vec[1]{e} and \vec[1]{e} that have been derived from \rnd during encryption.
	It then computes $\randomOracle(\mathsf{Supp}(\vec[1]{e},\vec[1]{e}))$ and outputs $\ctxt[1] \xor \randomOracle(\mathsf{Supp}(\vec[1]{e},\vec[1]{e}))$.
	\qed
\end{proof}
At this point, we would like to point out that the code-based PKE schemes which underlie the NIST proposals BigQuake~\cite{NISTPQC-R1:BIGQUAKE17}, HQC~\cite{NISTPQC-R2:HQC19}, and RQC~\cite{NISTPQC-R2:RQC19} are recoverable as well.

\subsubsection{QS2 Attack against ROLLO-II.}

We give an explicit attack against the \qINDqCPA security of ROLLO-II.
It 
is a Hadamard distinguisher that
exploits the fact that the message is essentially encrypted using a one-time pad (ciphertext part \ctxt[1] in Fig.~\ref{fig:ROLLOII}).
\begin{theorem}\label{thm:ROLLOIIattack}
	Let $\Sigma$ be the code-based PKE scheme ROLLO-II shown in Fig.~\ref{fig:ROLLOII}.
	Then there exists an efficient distinguishing adversary \Eve that wins the experiment $\gamefont{qIND\mbox{-}qCPA}(\Sigma,\Eve)$ with probability $\frac{3}{4}$.
\end{theorem}
\begin{proof}
	In the challenge phase, the adversary \Eve prepares the two states $\phi_0 = \ket{0\dots0}$ and $\phi_1 = \sum_{\mssg}\frac{1}{\sqrt{2^{|\mssg|}}}\ket{\mssg}$ and sends them to the challenger.
	Upon receiving the challenge ciphertext $\psi$, \Eve traces out the register \ket{\ctxt[2]} and measures the resulting state in the Hadamard basis.
	If the measurement outcome is $0$, \Eve outputs $1$, otherwise, it outputs $0$.
	
	If the secret bit $b$ is $1$, the state $\phi_1$ is encrypted.
	Then the state $\psi$ is
	\begin{align*}
		O^{(2)}_\Enc \left( \sum_{\mssg} \frac{1}{\sqrt{2^{|\mssg|}}} \ket{\mssg} \right) 
		= \sum_{\mssg} \frac{1}{\sqrt{2^{|\mssg|}}} \ket{\Enc(\mssg)} 
		= \sum_{\mssg} \frac{1}{\sqrt{2^{|\mssg|}}} \ket{\ctxt_1^{(\mssg)}} \otimes \ket{\ctxt_2} 
		\,;
	\end{align*}
	where $\ctxt_1^{(\mssg)}$ is the $\ctxt_1$ part of the ciphertext (cf.~Fig.~\ref{fig:ROLLOII}) related to superposition element $\mssg$, while the second part $\ctxt_2$ is independent of the underlying plaintext.
	Hence, the two corresponding registers are unentangled, and after tracing out the second, \Eve gets the state
	\begin{align*}
		\sum_{\mssg} \frac{1}{\sqrt{2^{|\mssg|}}} \ket{\ctxt_1^{(\mssg)}} = 
		\sum_{\mssg} \frac{1}{\sqrt{2^{|\mssg|}}} \ket{\mssg \xor \randomOracle(\mathsf{E})} =
		\sum_{\mssg} \frac{1}{\sqrt{2^{|\mssg|}}} \ket{\mssg} =
		\ket{+}
		\,;
	\end{align*}
	hence measuring in the Hadamard basis yields $0$ with probability $1$.
	
	If the secret bit $b$ is $0$, the state $\phi_0$ is encrypted and the outcome of the final Hadamard measurement will be a $0$ or $1$, each with probability $50\%$, which concludes the proof.
	\qed
\end{proof}
We note that the attack also works against the code-based scheme BigQuake~\cite{NISTPQC-R1:BIGQUAKE17} which uses the same one-time pad approach to encrypt the message.

\subsection{Results for Hybrid Encryption} \label{sec:SecurityAnalysisHybrid}

In this section we analyse the canonical hybrid encryption scheme with respect to our \qINDqCPA security notion.
We show that its security mainly depends on the underlying symmetric key encryption scheme.

\subsubsection{Hybrid Encryption Scheme.} \label{sec:HybridPKEscheme}

The canonical hybrid PKE-SKE encryption scheme combines a public key encryption and a symmetric key encryption scheme into a public key encryption scheme.
That is, a message is encrypted using a fresh one-time key of the symmetric encryption scheme.
The one-time key is then encrypted using the public key encryption scheme, whereupon the encrypted one-time key is attached to the ciphertext.
To decrypt, one first recovers the symmetric one-time key, and then uses it to decrypt the ciphertext containing the message.
The canonical hybrid encryption scheme is shown in Fig.~\ref{fig:PseudocodeHybridPKE}.
For additional background on symmetric key encryption schemes and the security notion used in this section, see Appendix~\ref{app:additionalBackground}.
\begin{figure}[htbp]
	\centering
			\begin{pchstack}
			\fbox{
				\procedure[width=11em]{$\kgen(\secpar)$}{%
					\x (\pk,\sk) \getsr \kgen^P(\secpar) \\
					\x \pcreturn (\pk,\sk)
				}
				\pchspace
				\procedure[width=11em]{$\enc_\pk(\mssg;\rnd)$}{%
					\x \pcparse \rnd~\pcas (\rnd[1],\rnd[2],\rnd[3]) \\
					\x \key \gets \kgen^S(\secpar;\rnd[1]) \\
					\x \ctxt[1] \gets \enc^S_{\key}(\mssg;\rnd[2]) \\
					\x \ctxt[2] \gets \enc^P_\pk(\key;\rnd[3]) \\
					\x \pcreturn (\ctxt[1],\ctxt[2])
				}		
				\pchspace
				\procedure[width=11em]{$\dec_\sk(\ctxt)$}{%
					\x \pcparse \ctxt~\pcas (\ctxt[1],\ctxt[2]) \\
					\x \key \gets \dec^P_\sk(\ctxt[2]) \\
					\x \mssg \gets \dec^S_{\key}(\ctxt[1]) \\
					\x \pcreturn \mssg
				}
			}
		\end{pchstack}	
		
	\caption{Hybrid encryption scheme $\Sigma = (\kgen,\enc,\dec)$ built from a PKE scheme $\Sigma^P = (\kgen^P,\enc^P,\dec^P)$ and an SKE scheme $\Sigma^S = (\kgen^S,\enc^S,\dec^S)$.}
	\label{fig:PseudocodeHybridPKE}
\end{figure}

Below we show that the canonical hybrid encryption scheme is recoverable.
Given the randomness, the used one-time key can be obtained, which allows to decrypt the ciphertext part that contains the message.
We emphasise that the hybrid encryption scheme is recoverable even if the underlying PKE scheme is not recoverable.
\begin{lemma}\label{lem:hybridPKErecoverable}
	The canonical hybrid encryption scheme $\Sigma = (\kgen,\enc,\dec)$, shown in Fig.~\ref{fig:PseudocodeHybridPKE}, is recoverable as from Definition~\ref{def:recoverablePKEscheme}.
\end{lemma}
\begin{proof}
	To prove the statement, we have to specify the algorithm \Rec that is introduced in Definition~\ref{def:recoverablePKEscheme}.
	Its input is a public key $\pk$, a randomness \rnd, and a ciphertext $\ctxt = (\ctxt[1],\ctxt[2])$, such that $\ctxt$ corresponds to the encryption of a message \mssg, using the public key \pk and randomness \rnd.
	The algorithm \Rec proceeds as follows.
	Given the randomness \rnd, it obtains \rnd[1], \rnd[2], and \rnd[3], which have been derived from \rnd during encryption.
	It then computes $\key \gets \kgen^S(\secpar;\rnd[1])$ and outputs $\dec^S_{\key}(\ctxt[1])$.
	This concludes the proof.
	\qed
\end{proof}

\subsubsection{\qINDqCPA Security of Hybrid Encryption.} \label{sec:HybridPKE_QS2Security}

We now turn our attention towards the QS2 security of the hybrid encryption scheme.
It turns out that the QS2 security depends on the underlying SKE scheme, while the underlying PKE scheme merely requires QS1 security.
This is formalised in the theorem below.
\begin{theorem}[QS2 Security of Hybrid Encryption]\label{thm:hybridEncryptionQS2security}
	Let $\Sigma =  (\kgen,\enc,$ $\dec)$ be the hybrid encryption scheme built from an SKE scheme $\Sigma^S = 
	(\kgen^S,$ $\enc^S,\dec^S)$ and a PKE scheme $\Sigma^P = (\kgen^P,\enc^P,$ $\dec^P)$, as shown in 
	Fig.~\ref{fig:PseudocodeHybridPKE}.
	For any adversary \Eve against $\Sigma$, there exist adversaries \EveB and \EveC against $\Sigma^{P}$ and $\Sigma^{S}$, respectively, such that
	\begin{align*}
		\Adv[\Sigma]{\qINDqCPA}{\Eve} \le \Adv[\Sigma^{P}]{\INDqCPA}{\EveB} + \Adv[\Sigma^{S}]{\qIND}{\EveC}\,.
	\end{align*}
\end{theorem}
\begin{proof}
	The proof uses two games \game[0] and \game[1], where \game[0] is the \qINDqCPA security game instantiated with $\Sigma$, and \game[1] is the same except that the ciphertext part \ctxt[2] is replaced by encrypting a random key $\key'$ rather than \key.
	It holds that
	\begin{align*}
		\Adv[\Sigma]{\qINDqCPA}{\Eve} = \gameAdv{\game[0]^{\Eve}}{\game[1]^{\Eve}} + \Adv{\game[1]}{\Eve}\,.
	\end{align*}
	We construct the following adversary \EveB which receives a public key \pk as input.
	It flips a bit $b$ at random and runs \Eve on the same public key \pk.
	It answers every learning query $\phi$ by \Eve by generating a one-time key \key, asking its own challenger for an encryption of this key to obtain the ciphertext \ket{\ctxt[2]}, locally computes \ket{\ctxt[1]} by applying the type-2 encryption operator $U_{\enc^{S}}^{(2)}$ to $\phi$, and sends the ciphertext back to \Eve.
	For the challenge query $\phi_{0},\phi_{1}$ by \Eve, \EveB picks two (classical) keys \key and $\key'$, applies the type-2 encryption operator $U_{\enc_{\key}^{S}}^{(2)}$, using key \key, to $\phi_b$ to obtain \ket{\ctxt[1]}, obtains \ket{\ctxt[2]} by sending \key and $\key'$ to its own challenger, and sends the ciphertext back to \Eve.
	When \Eve guesses the bit $b$ correctly, \EveB outputs $0$, otherwise, it outputs $1$.
	It holds that \EveB perfectly simulates \game[0] and \game[1], depending on its own challenge, hence
	\begin{align*}
		\gameAdv{\game[0]^{\Eve}}{\game[1]^{\Eve}} \le \Adv[\Sigma^{P}]{\INDqCPA}{\EveB}\,.
	\end{align*}
	Next we transform an adversary \Eve, playing \game[1], into an \qIND adversary \EveC against $\Sigma^{S}$.
	The adversary \EveC generates a key pair (\pk,\sk) for the underlying PKE, which allows to perform all operations related to the PKE scheme.
	It runs \Eve on the public key \pk and answers any learning query by generating a key \key which it uses to encrypt the query by \Eve and then encrypts this key using the PKE scheme.
	The challenge query $\phi_{0},\phi_{1}$ by \Eve is forwarded by \EveC as its own challenge to obtain the ciphertext \ket{\ctxt[2]}, while the ciphertext \ket{\ctxt[1]} is computed locally by encrypting a randomly generated key \key using the PKE scheme.
	When \Eve outputs its guess, \EveC forwards it as its own output.
	It holds that \EveC perfectly simulated game \game[1] for \Eve, with the same secret bit as its \qIND security game, thus it holds that
	\begin{align*}
		\Adv{\game[1]}{\Eve} \le \Adv[\Sigma^{S}]{\qIND}{\EveC}\,.
	\end{align*}
	Collecting the bounds above proves the statement.
	\qed
\end{proof}
Theorem~\ref{thm:hybridEncryptionQS2security} reveals that to achieve our \qINDqCPA security notion, we can instantiate the hybrid encryption scheme with a PKE that merely achieves QS1 security.
This allows the usage of ROLLO-II, which, used as a stand-alone PKE scheme, is not \qINDqCPA-secure.

In the following we show that Theorem~\ref{thm:hybridEncryptionQS2security} is strict.
If the underlying SKE is not \qIND-secure, then the resulting hybrid scheme is not \qINDqCPA-secure, irrespectively of the underlying PKE scheme.
This is shown in the theorem below.
Examples for SKE schemes which are not \qIND-secure are given in~\cite{C:GagHulSch16}.
\begin{theorem}\label{thm:hybridEncryptionQS2insecurity}
	Let $\Sigma =  (\kgen,\enc,\dec)$ be the hybrid encryption scheme built from an SKE scheme $\Sigma^S = 
	(\kgen^S,\enc^S,\dec^S)$ and a PKE scheme $\Sigma^P = (\kgen^P,\enc^P,$ $\dec^P)$, as shown in 
	Fig.~\ref{fig:PseudocodeHybridPKE}.
	Assume that there exists an adversary \Eve which has some non-negligible advantage $\epsilon$ against the \qIND security of $\Sigma^S$.
	Then there exists an adversary \EveB against $\Sigma$ such that
	\begin{align*}
		\Adv[\Sigma]{\qINDqCPA}{\EveB} \ge \epsilon\,.
	\end{align*}
\end{theorem}
\begin{proof}
	We construct the adversary \EveB, which uses adversary \Eve as subroutine, as follows.
	When \Eve outputs its challenge messages $\phi_0$ and $\phi_1$, \EveB forwards these to its own challenger.
	Upon receiving the challenge ciphertext $\psi = \ket{\ctxt[1]}\tensor\ket{\ctxt[2]}$, \EveB sends \ket{\ctxt[1]} to \Eve.
	When \Eve outputs its guess $b'$, \EveB outputs $b'$ as its own guess.
	It holds that \EveB perfectly simulates the \qIND security experiment, with the same challenge bit $b$, for \Eve.
	By outputting the same bit as \Eve, we have
	\begin{align*}
		\Adv[\Sigma]{\qINDqCPA}{\EveB} \ge \Adv[\Sigma^{S}]{\qIND}{\Eve} = \epsilon\,,
	\end{align*}
	which proves the claim.
	\qed
\end{proof}

\subsection{Discussion} \label{sec:SecurityAnalysisDiscussion}

In this section, we gave both positive and negative examples regarding the QS2 security of real-world public key encryption schemes.
{We gave a concrete attack against the canonical LWE-based scheme for the case $q=2$ and an attack against the code-based scheme ROLLO-II showing that these schemes are \qINDqCPA insecure.}
These results, however, considered that the scheme are used as public key encryption schemes to encrypt the actual message.
On the other hand, Theorem~\ref{thm:hybridEncryptionQS2security} reveals that both ROLLO-II and the canonical LWE-based scheme are sufficient to achieve \qINDqCPA secure when used as a key encapsulation mechanism, together with a QS2-secure SKE scheme.

The standardisation effort by NIST focuses on the latter scenario, hence our results show that for these standardised schemes it is sufficient to achieve QS1 security in order for the resulting KEM to achieve our stronger, more conservative security notion.
At the same time, our results also show that extra cautiousness is necessary when these standardised schemes are deployed directly as PKE schemes in protocols that require security in the QS2 sense.

\section{Classifying Other Public Key Encryption Schemes} \label{sec:classif}

So far we have built a framework for QS2 security of PKE schemes which are perfectly correct or recoverable (or both). 
But what about schemes which do not fall in either of these two categories? 
Are there such examples at all?
And what can we learn from this?
In this section, we initiate the classification of PKE schemes in general, extend our results to other classes of PKE schemes where possible, and point out the obstacles in other cases.

\subsection{Dealing with Decryption Failures: The General Case}\label{sec:decrfailures}

First, we discuss why arbitrary non-correct PKE schemes do not allow, in general, to define a type-2 encryption operator and,
consequently, we cannot 
always 
define the \qINDqCPA game as from Experiment~\ref{exp:qINDqCPA}.
However, we also discuss a possible workaround.

First of all, recall that defining a type-2 operator is only possible for functions that are inherently invertible. 
Then observe that a $(1-\alpha)$-correct PKE scheme (cf. Definition~\ref{def:partiallyCorrectPKE}) could have arbitrary, even overwhelming decryption error $\alpha$. 
In the most extreme case, the scheme can be almost identical to a constant function (for example, consider an artificial scheme where every public key \pk always encrypts to $0$, except for one particular randomness $\bar{\rnd}$ where it produces a correctly decryptable ciphertext instead).
In the presence of decryption failures, it is therefore impossible to find a general way to define type-2 operators for encryption, and hence, to define a suitable \qINDqCPA security notion.\footnote{Recoverable schemes are a special case: they might not be always reversible in the message space only, but they are always reversible in the union of message space and randomness space.}

We call {\em non-isometric} such schemes, where it is simply not possible to define a unitary operator that behaves {\em exactly} as from Definition~\ref{def:type2enc} for any keypair, even if we drop the requirement of efficiency.
\begin{definition}[Non-Isometric Schemes]\label{def:nonisometric}
	Let $\Sigma =(\kgen,\enc,\dec)$ be a PKE scheme. 
	We say that $\Sigma$ is {\em non-isometric} if, for any $(\pk,\sk) \getsr \KGen$,
	there exists at least a 
	randomness $\rnd_\pk$ such that the function $\mssg \mapsto \Enc_\pk(\mssg;\rnd_\pk)$ is non-injective.
	In particular, for any unitary $U$ acting on the 
	appropriate Hilbert spaces,
	there exists at least a 
	pair $(\mssg_\pk,\rnd_\pk)$ 
	such that:
	\begin{equation*}
		\Pr \left[ M(U\ket{\rnd_\pk,\mssg_\pk, 0,\ldots ,0}) \to \left( \rnd_\pk , \Enc_\pk(\mssg_\pk;\rnd_\pk) \right) \right] < 1 \,,
	\end{equation*}
	where $M$ denotes measurement in the canonical computational base.
\end{definition}
A possible workaround for these non-isometric schemes is to `enforce' the reversibility of the encryption, obtaining a new type of encryption unitary.
Consider what happens if we want to use the type-1 encryption operator (cf. Definition~\ref{def:type1enc}) during the challenge query:
\begin{equation*}
	U_{\Enc_\pk}^{(1)} : \ket{\rnd,\mssg,y} \mapsto \ket{\rnd,\mssg,y\xor\enc_\pk(\mssg;\rnd)}\,.
\end{equation*}
As already observed, the randomness \rnd can be understood as classical and discarded by the challenger.
However, the other two registers are generally going to be entangled, and both would have to be sent to the adversary for a meaningful quantum notion; but this would clearly break security because the message would remain in clear.\footnote{This explanation appears in detail in~\cite{C:GagHulSch16}.}
We could try to `fix' this issue by (reversibly) masking the message register sent to the adversary, for example by using a permutation $\pi$ on the message space drawn uniformly at random.
The following unitary:
\begin{equation*}
	U_{\Enc_\pk}^{(\pi)} : \ket{\rnd,\mssg,y} \mapsto \ket{\rnd,\pi(\mssg),y\xor\enc_\pk(\mssg;\rnd)}\,
\end{equation*}
allows hence to define a new type of quantum challenge query, where the challenger still discards the randomness 
register after encryption, but sends back the other two registers to the adversary.
Notice how, from the adversary's point of view, $\pi(\mssg)$ is a completely random element, and therefore the presence of this additional register does not offer any distinguishing advantage.
Moreover, in actual security reductions, the uniformly drawn $\pi$ can be replaced by a quantum-secure pseudorandom permutation~\cite{C:GagHulSch16}, or \QPRP in short.

We can hence use these {\em type-$\pi$ operators} to define (for \emph{any} PKE scheme, including the non-isometric ones) a new indistinguishability game and a related security notion with quantum challenge query. 
Motivating the use of such operators when modelling security is arguably non-trivial.
In certain cases, one could see $\pi(\mssg)$ as some sort of side-channel information given to the adversary, but in general it looks like just an artificial way to enforce reversibility on schemes which are not. 
We will therefore not study the resulting security notion in this work, but we want nevertheless to make a few 
observations on it.

First of all, notice that such a new security notion cannot be stronger than \qINDqCPA, at least when considering perfectly correct or recoverable schemes.
As a separating example, consider the distinguishing attack from Theorem~\ref{thm:hybridEncryptionQS2insecurity}: this will not work any more because of the presence of the entangled $\pi(\mssg)$ register, so that the hybrid scheme might be secure according to the new notion but still \qINDqCPA insecure.

Second, notice how the challenge query resulting from the use of type-$\pi$ operators reminds of the one given in an alternative quantum indistinguishability notion for secret key encryption schemes proposed by Mossayebi and Schack~\cite{MossayebiS16} - the difference is basically producing $\ket{\mssg, \Enc_\pk(\pi(\mssg))}$ instead of $\ket{\pi(\mssg), \Enc_\pk(\mssg)}$ - which is itself not comparable to \qINDqCPA. 
This security notion has been recently investigated and expanded by Chevalier et al.~\cite{eprint:ChevalierEV20}.

\subsection{Refining the Classification}\label{sec:refining}

Now we know how to define \qINDqCPA security of PKE schemes which are perfectly correct or recoverable (or both), 
and at the same time we know that it is not possible for schemes that are non-isometric. 
But it turns out we can say more. 
First of all we make a distinction for those schemes which {\em are} isometric: it means that it is possible to define a unitary operator that behaves exactly as a \mbox{type-2} encryption operator, but we distinguish whether finding and building such operator is efficient or not.
\begin{definition}[(Efficiently) Isometric Schemes]\label{def:isometric}
	Let $\Sigma$ be a PKE scheme with $\Sigma =(\kgen,\enc,\dec)$. 
	We say that $\Sigma$ is {\em isometric} if, 
	for any $(\pk,\sk) \getsr \KGen$ 
	and 
	for any randomness $\rnd$ the function $\mssg \mapsto \Enc_\pk(\mssg;\rnd)$ is injective. In particular, there exists a unitary $U$ acting on the 
	appropriate Hilbert spaces,
	such that for any $(\mssg,\rnd)$:
	\begin{equation*}
		\Pr \left[ M(U\ket{\rnd,\mssg, 0,\ldots ,0}) \to \left( \rnd , \Enc_\pk(\mssg;\rnd) \right) \right] = 1 \,,
	\end{equation*}
	where $M$ denotes measurement in the canonical computational base. 
	Furthermore, we say that $\Sigma$ is {\em efficiently isometric} if $U$ can be efficiently realised.
\end{definition}
Notice how, in general, an isometric scheme is not necessarily efficiently isometric. This is because, unlike for type-1 operators, the efficiency of the \Enc procedure is only enough to guarantee the existence of a unitary $U$ with the above property, but not its {\em efficiency}.
Then, notice how a type-2 encryption operator (as from Definition~\ref{def:type2enc}) satisfies the above definition of $U$, both by construction and by efficiency. 
In other words, efficiently isometric schemes are exactly all and only those schemes which, by definition, admit an efficient construction of the type-2 encryption operator. 
Clearly, in particular this includes perfectly correct schemes (by Theorem~\ref{thm:type2enceff}) and recoverable schemes (by Theorem~\ref{thm:recoverabletype2}).
\begin{corollary}
	Let $\Sigma$ be a PKE scheme. If $\Sigma$ is perfectly correct or recoverable, then it is efficiently isometric.
\end{corollary}
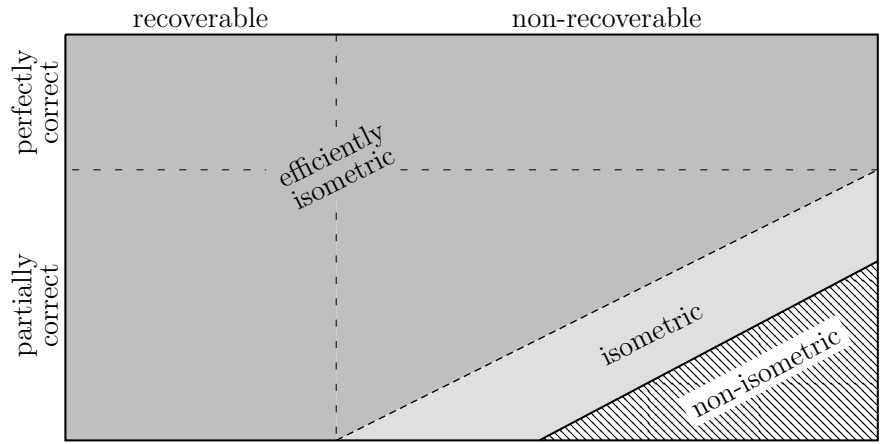
\begin{figure}[htbp]
	\centering
	\scalebox{0.9}{
		\begin{tikzpicture}
			\coordinate (LEFTLOWER) at (0,0);
			\coordinate (RIGHTLOWER) at (12,0);
			\coordinate (LEFTUPPER) at (0,6);
			\coordinate (RIGHTUPPER) at (12,6);
			
			\coordinate (INTERSECTIONTOP) at (4,6);
			\coordinate (INTERSECTIONLEFT) at (0,4);
			\coordinate (INTERSECTIONRIGHT) at (12,4);
			\coordinate (INTERSECTIONBOTTOM) at (4,0);
			\coordinate (INTERSECTIONMIDDLE) at (4,4);

			\draw ($(LEFTUPPER) + (0.0,0.0)$) -- ($(INTERSECTIONTOP) + (0.0,0.0)$) node [midway, above, sloped,align=center] (TextNode) {\large recoverable};
			\draw ($(INTERSECTIONTOP) + (0.0,0.0)$) -- ($(RIGHTUPPER) + (0.0,0.0)$) node [midway, above, sloped,align=center] (TextNode) {\large non-recoverable};
			\draw ($(INTERSECTIONLEFT) + (0.0,0.0)$) -- ($(LEFTUPPER) + (0.0,0.0)$) node [midway, above, sloped,align=center] (TextNode) {\large perfectly\\\large correct};
			\draw ($(LEFTLOWER) + (0.0,0.0)$) -- ($(INTERSECTIONLEFT) + (0.0,0.0)$) node [midway, above, sloped,align=center] (TextNode) {\large partially\\\large correct};

			\path[fill=gray!50] (LEFTLOWER) -- (INTERSECTIONLEFT) -- (INTERSECTIONMIDDLE) -- (INTERSECTIONBOTTOM) -- (LEFTLOWER);

			\path[fill=gray!50] (INTERSECTIONTOP) -- (RIGHTUPPER) -- (INTERSECTIONRIGHT) -- (INTERSECTIONMIDDLE) -- (INTERSECTIONTOP);

			\path[fill=gray!50] (LEFTUPPER) -- (INTERSECTIONTOP) -- (INTERSECTIONMIDDLE) -- (INTERSECTIONLEFT) -- (LEFTUPPER);

			\path[fill=gray!50] (INTERSECTIONMIDDLE) -- (INTERSECTIONRIGHT) -- (INTERSECTIONBOTTOM) -- (INTERSECTIONMIDDLE);

			\path[] (LEFTLOWER) -- (LEFTUPPER)  -- (RIGHTUPPER)  -- (INTERSECTIONRIGHT)  -- (INTERSECTIONBOTTOM)  -- (LEFTLOWER); 

			\path[fill=gray!25] (INTERSECTIONBOTTOM) -- (INTERSECTIONRIGHT) -- ($(INTERSECTIONRIGHT) + (0.0,-1.35)$) -- ($(INTERSECTIONBOTTOM) + (3.0,0.0)$) -- (INTERSECTIONBOTTOM); 
			\draw[line width = 0.2mm, densely dashed] (INTERSECTIONBOTTOM) -- (INTERSECTIONRIGHT) node [midway, below, sloped, xshift=0.4cm,yshift=-0.4cm, fill=gray!25] (TextNode) {\large isometric};

			\path[pattern = north west lines, dashed] (RIGHTLOWER) -- ($(INTERSECTIONBOTTOM) + (3.0,0.0)$) -- ($(INTERSECTIONRIGHT) + (0.0,-1.35)$) -- (RIGHTLOWER); 
			\draw[line width = 0.3mm] ($(INTERSECTIONBOTTOM) + (3.0,0.0)$) -- ($(INTERSECTIONRIGHT) + (0.0,-1.35)$) node [midway, below, sloped, xshift=0.6cm,yshift=-0.4cm, fill=white] (TextNode) {\large non-isometric};

			\draw[line width = 0.3mm] (LEFTUPPER) -- (RIGHTUPPER) -- (RIGHTLOWER) -- (LEFTLOWER) -- (LEFTUPPER);
			\draw[loosely dashed, line width = 0.15mm] (INTERSECTIONMIDDLE) -- (INTERSECTIONTOP);
			\draw[loosely dashed, line width = 0.15mm] (INTERSECTIONMIDDLE) -- (INTERSECTIONLEFT);
			\draw[loosely dashed, line width = 0.15mm] (INTERSECTIONMIDDLE) -- (INTERSECTIONBOTTOM);
			\draw[loosely dashed, line width = 0.15mm] (INTERSECTIONMIDDLE) -- (INTERSECTIONRIGHT);
			
			\path[] (LEFTLOWER) -- (RIGHTUPPER) node [midway, sloped, xshift=-1.30cm,yshift=1.85cm, fill=gray!50] (TextNode)[align=center] {\large efficiently\\\large isometric};
		\end{tikzpicture}
	}
	\caption{Classification of PKE schemes.
		The \qINDqCPA security notion can be defined for all schemes except those in the shaded area (non-isometric). For efficiently isometric schemes (dark gray area) the \mbox{type-2} operator can be realised efficiently and we provide concrete circuits for the schemes that are perfectly correct or recoverable. For isometric schemes (light gray area) the \mbox{type-2} operator can be realised but not efficiently.}
	\label{fig:PKEclassification}
\end{figure}
The situation is depicted in Fig.~\ref{fig:PKEclassification}. This means that, as from Definition~\ref{def:qINDqCPA}, we can extend the \qINDqCPA security notion not only to recoverable or perfectly correct schemes, but to all the efficiently isometric ones.
For the non-efficient case (arbitrary isometric schemes) the \qINDqCPA notion can still be defined, but its usefulness would be less clear, as it might require unbounded challengers in the security game (and therefore, difficulty in simulating them by efficient reductions when proving the security of a particular scheme). Still, it would be useful for {\em impossibility results}, i.e., proving that a particular isometric scheme is not \qINDqCPA-secure.

Finally, can we find representative examples of schemes which fall in the categories that we have just defined?
We have already mentioned an example of a non-isometric scheme at the beginning of Section~\ref{sec:decrfailures} (the almost-constant one).
Here we show a construction of an efficiently isometric scheme that is neither perfectly correct nor recoverable. 
The construction is given in Fig.~\ref{fig:patchedEncryptionScheme}: 
it transforms a recoverable, not perfectly correct encryption scheme by pre-processing the message with a quantum-secure trapdoor permutation~\cite{DBLP:phd/dnb/Gagliardoni17} before encrypting it, and inverts again the permutation after decryption (the public and secret keys of the trapdoor permutation are embedded in the public and secret key, respectively, of the resulting scheme).
It works because the permutation `scrambles' the resulting ciphertexts but not the randomness, thereby hindering an adversary (or a challenger) who tries to build an efficient recovery algorithm \Rec for the transformed scheme.
At the same time, we show how such construction is efficiently isometric, by showing an efficient circuit for the canonical type-2 encryption operator $U^{(2)}_\Enc$.
This is formalised in the theorem below.

\begin{figure}[htbp]
	\centering
			\begin{pchstack}
			\fbox{
				\procedure[width=11em]{$\kgen(\secpar)$}{%
					\x (\pk_{e},\sk_{e}) \getsr \kgen^{\Sigma}(\secpar) \\
					\x (\pk_{f},\sk_{f}) \getsr \kgen^{F}(\secpar) \\
					\x \pk \gets (\pk_{e},\pk_{f}) \\
					\x \sk \gets (\sk_{e},\sk_{f}) \\
					\x \pcreturn (\pk,\sk)
				}
				\pchspace
				\procedure[width=11em]{$\enc_\pk(\mssg;\rnd)$}{%
					\x \pcparse \pk~\pcas (\pk_{e},\pk_{f}) \\
					\x y \gets \TDF(\pk_{f},\mssg) \\
					\x \ctxt \gets \enc^{\Sigma}(\pk_{e},y;\rnd) \\
					\x \pcreturn \ctxt
				}		
				\pchspace
				\procedure[width=11em]{$\dec_\sk(\ctxt)$}{%
					\x \pcparse \sk~\pcas (\sk_{e},\sk_{f}) \\
					\x y \gets \dec^{\Sigma}(\sk_{e},\ctxt) \\
					\x \mssg \gets \TDFinv(\sk_{f},y) \\
					\x \pcreturn \mssg
				}
			}
		\end{pchstack}	
		
	\caption{Transformed scheme $\Gamma$, where $\Sigma = (\kgen^{\Sigma},\enc^{\Sigma},\dec^{\Sigma})$ is a PKE scheme and $\Pi = (\kgen^{F},\TDF,\TDFinv)$ is a deterministic trapdoor permutation.}
	\label{fig:patchedEncryptionScheme}
\end{figure}

\begin{theorem}\label{thm:patch1}
	Let $\Pi = (\kgen^{F},\TDF,\TDFinv)$ be a deterministic trapdoor permutation and $\Sigma = (\kgen^{\Sigma},\enc^{\Sigma},\dec^{\Sigma})$ be a PKE scheme.
	If $\Pi$ is quantum-secure and $\Sigma$ is recoverable and $(1-\alpha)$-correct, 
	then the scheme $\Gamma = (\kgen,\enc,\dec)$ depicted in Fig.~\ref{fig:patchedEncryptionScheme} is $(1-\alpha)$-correct, non-recoverable, and efficiently isometric PKE.
\end{theorem}
\begin{proof}
	Partial correctness of the encryption scheme $\Gamma$ follows immediately from the partial correctness of $\Sigma$, as permuting the messages does not change the overall decryption failure probability.
	
	Assume, for sake of contradiction, that $\Gamma$ is recoverable.
	Then there exists an efficient algorithm \Rec that, on input $\pk$, $\rnd$, and $\enc_{\pk}(\mssg;\rnd)$, outputs \mssg. 
	We construct the following adversary \EveB against the trapdoor permutation $\Pi$.
	He receives a public key $\pk_{f}$ for the trapdoor permutation \TDF along with $y \gets \TDF_{\pk_{f}}(x)$ for a random $x$, and is asked to find $x$.
	\EveB computes $(\pk_{e},\sk_{e}) \getsr \kgen^{\Sigma}(\secpar)$, chooses $\rnd \getsr \rndsp$, computes $\ctxt \gets \enc^{\Sigma}_{\pk_{e}}(y;\rnd)$, 
	and uses $\pk = (\pk_{e},\pk_{f}),\rnd,\ctxt$ as an input to \Rec.
	By construction, we have $\ctxt = \enc^{\Sigma}_{\pk_{e}}(\TDF_{\pk_{f}}(x);\rnd) = \enc_\pk(x;\rnd)$, hence $\Rec$ outputs $x$.
	So \EveB can find the correct preimage with probability $1$, hence breaking the security of the trapdoor permutation. 
	This contradicts the recoverability of $\Gamma$.
	
	Finally, in Fig.~\ref{fig:patchisom} we show an efficient circuit for the realisation of $U^{(2)}_\Enc$. This uses subcircuits for computing the type-1 operator for the trapdoor permutation and its inverse (given the trapdoor permutation's public key and secret key), and type-1 encryption and recovery for the underlying PKE scheme.
	\qed
\end{proof}

\begin{figure}[htbp]
	\centering
	\scalebox{0.9}{
		\begin{tikzpicture}
			
			\node (MIDDLE) [text width=1em,align=center] {};
			\node (MIDDLETRAPDOOR) [text width=1em,align=center] [left of=MIDDLE, node distance = 3.0cm] {};
			\node (MIDDLEENCRYPTION) [text width=1em,align=center] [right of=MIDDLE, node distance = 3.0cm] {};
			
			\node (UF)[text width=1cm,align=center,fill=gray!25,minimum height=2cm,minimum width=1.5cm,draw] [left of=MIDDLETRAPDOOR, node distance = 1.25cm, yshift=-0.5cm] {\large $U_{\TDF}^{(1)}$};
			\node (UFinv)[text width=1cm,align=center,fill=gray!25,minimum height=2cm,minimum width=1.5cm,draw] [right of=MIDDLETRAPDOOR, node distance = 1.25cm, yshift=-0.5cm] {\large $U_{\TDFinv}^{(1)}$};
			
			\node (Uenc)[text width=1cm,align=center,fill=gray!25,minimum height=3.0cm,minimum width=1.5cm,draw] [left of=MIDDLEENCRYPTION, node distance = 1.25cm] {\large $U_{\Enc}^{(1)}$};
			\node (Urec)[text width=1cm,align=center,fill=gray!25,minimum height=3.0cm,minimum width=1.5cm,draw] [right of=MIDDLEENCRYPTION, node distance = 1.25cm] {\large $U_{\Rec}^{(1)}$};
			
			\node (IN1) [text width=0.5cm,align=center] [left of=MIDDLETRAPDOOR, yshift=1.0cm, node distance=3.5cm] {\large \ket{\rnd}};
			\node (IN2) [text width=0.5cm,align=center] [left of=MIDDLETRAPDOOR, yshift=0.0cm, node distance=3.5cm] {\large \ket{\mssg}};
			\node (IN3) [text width=0.4cm,align=center] [left of=MIDDLETRAPDOOR, yshift=-1.0cm, node distance=2.5cm] {\large \ket{0}};
			
			\node (OUT1) [text width=0.5cm,align=center] [right of=MIDDLEENCRYPTION, yshift=1.0cm, node distance=3.5cm] {\large \ket{\rnd}};
			\node (OUT2) [text width=0.5cm,align=center] [right of=MIDDLEENCRYPTION, yshift=0.0cm, node distance=3.5cm] {\large \ket{\ctxt}};
			\node (OUT3) [text width=0.4cm,align=center] [right of=MIDDLEENCRYPTION, yshift=-1.0cm, node distance=2.5cm] {\large \ket{0}};
			
			\node (OUTTRAPDOOR) [text width=0.4cm,align=center] [right of=MIDDLETRAPDOOR, yshift=-1.0cm, node distance=2.5cm] {\large \ket{0}};
			\node (INENCRYPTION) [text width=0.4cm,align=center] [left of=MIDDLEENCRYPTION, yshift=-1.0cm, node distance=2.5cm] {\large \ket{0}};
			
			\draw[] ($(IN1.east) + (0.0,0.0)$) -- ($(Uenc.west) + (0.0,1.0)$);
			\draw[] ($(Uenc.east) + (0.0,1.0)$) -- ($(Urec.west) + (0.0,1.0)$);
			\draw[] ($(Urec.east) + (0.0,1.0)$) -- ($(OUT1.west) + (0.0,0.0)$);
			
			\draw[] ($(IN2.east) + (0.0,0.0)$) -- ($(UF.west) + (0.0,0.5)$);
			\draw[] ($(IN3.east) + (0.0,0.0)$) -- ($(UF.west) + (0.0,-0.5)$);
			
			\draw ($(UF.east) + (0.0,0.5)$) to[out=0,in=-180] ($(UFinv.west) + (0.0,-0.5)$);
			\draw ($(UF.east) + (0.0,-0.5)$) to[out=0,in=-180] ($(UFinv.west) + (0.0,0.5)$);
			
			\draw[] ($(UFinv.east) + (0.0,0.5)$) -- ($(Uenc.west) + (0.0,0.0)$);
			\draw[] ($(UFinv.east) + (0.0,-0.5)$) -- ($(OUTTRAPDOOR.west) + (0.0,0.0)$);
			\draw[] ($(INENCRYPTION.east) + (0.0,0.0)$) -- ($(Uenc.west) + (0.0,-1.0)$);
			
			\draw ($(Uenc.east) + (0.0,0.0)$) to[out=0,in=-180] ($(Urec.west) + (0.0,-1.0)$);
			\draw ($(Uenc.east) + (0.0,-1.0)$) to[out=0,in=-180] ($(Urec.west) + (0.0,0.0)$);
			
			\draw[] ($(Urec.east) + (0.0,0.0)$) -- ($(OUT2.west) + (0.0,0.0)$);
			\draw[] ($(Urec.east) + (0.0,-1.0)$) -- ($(OUT3.west) + (0.0,0.0)$);
			
			\begin{scope}[on background layer]        
				\node (type2Box) [text width=3em,rounded corners=1ex,minimum height=3.75cm,minimum width=11.75cm,draw, densely dashed,line width = 0.3mm] at ($(MIDDLE)$) {};
			\end{scope}
		\end{tikzpicture}
	}
	\caption{Efficient realisation of the canonical type-2 encryption operator for the construction shown in Fig.~\ref{fig:patchedEncryptionScheme}.}
	\label{fig:patchisom}
\end{figure}
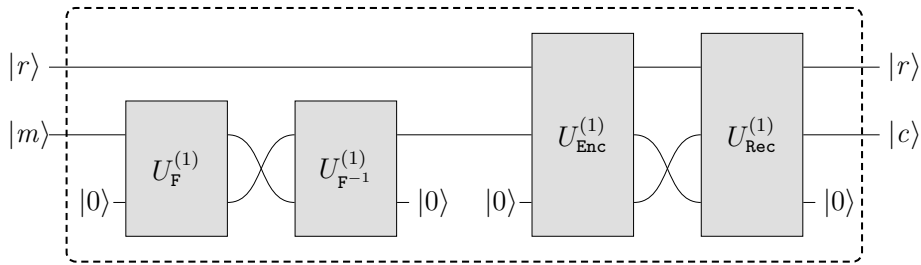
\begin{remark}
	Note that, albeit the above construction works at a theoretical level, there are currently no known candidates for quantum-secure trapdoor permutations.
	Alternatively, a quantum-secure \emph{injective trapdoor function} could be used instead, for which candidates exist.
	In this case, because of the inherent expansion factor, the message space for the transformed scheme will be smaller than the one in the original PKE scheme.
\end{remark}

\section{Future Directions}\label{sec:conclusions}

In this work we have filled the existing gap between the symmetric key and the public key case when defining security in the QS2 setting.
We showed how the existence of this gap was not due to a mere lack of interest, but because of non-trivial definitional issues that we solved.
We believe that our results provide useful guidelines in the security analysis of quantum-resistant PKE, but many research directions remain open to exploration.

In Section~\ref{sec:qINDqCCA} we sketch a general strategy for extending our results to the chosen ciphertext case. 
Although we believe that such a strategy works, we leave it as an open problem to formalise it correctly. 
We also leave it as an open problem to improve our game-based definitions to different provable security paradigms such as simulation-based.

We notice how our notions of \qINDqCPA for PKE can be also used to study the security of cryptographic primitives that `extend' PKE with extra functionalities. 
Such primitives include, for example, fully homomorphic encryption~\cite{STOC:Gentry09, C:BroJef15}, identity-based encryption~\cite{C:Zhandry12}, and functional encryption~\cite{TCC:BonSahWat11}.

We did not found any natural example of a scheme that is isometric, yet not efficiently so.
A simple idea would be to modify the construction from Fig.~\ref{fig:patchedEncryptionScheme} in such a way that the circuit provided in Fig.~\ref{fig:patchisom} becomes non-efficient (for example by using a hard to invert permutation instead of a trapdoor permutation).
This idea does not work for two reasons.
First, it would only show that \emph{this} particular construction of the type-2 operator is inefficient, while we would need to show that \emph{any} construction is.
Second, and more importantly, switching to a hard to invert permutation would make the decryption algorithm inefficient.
Hence the resulting scheme would no longer be a PKE scheme according to Definition~\ref{def:PKEscheme}.

Also, notice the following: given that \qINDqCPA is a stronger notion than \INDqCPA, having a PKE scheme where it is not even possible to define a type-2 encryption operator can actually be {\em desirable}. 
For such a scheme in fact, one should not worry about proving the (stricter) \qINDqCPA security notion, because the related attack scenario is simply not enforceable, and hence the scheme cannot be broken in a \qINDqCPA sense. 
So it would be interesting to find schemes which are \INDqCPA secure but non-isometric. 
We conjecture that a generic transformation to obtain such schemes is possible assuming the existence of quantum-secure {\em indistinguishability obfuscation}, but leave the problem open to further study.

We have also left unstudied the possibility of extending QS2 security notions to the use of type-$\pi$ operators, and to models where the adversary can query oracles on superpositions of public keys.

\subsection*{Acknowledgements}

The authors are very grateful to the anonymous reviewers for spotting a flaw in a previous version of this manuscript. 
The authors also thank Cecilia Boschini and Marc Fischlin for helpful discussions regarding the correctness of public key encryption schemes and Andreas Hülsing for general discussions on the content of this work.
TG acknowledges support by the EU H2020 Project FENTEC (Grant Agreement \#780108). 
JK and PS acknowledge funding by the Deutsche Forschungsgemeinschaft (DFG) – SFB 1119 – 236615297.

\addcontentsline{toc}{section}{References}
\bibliographystyle{abbrv}
\bibliography{cryptobib/abbrev3,cryptobib/crypto,local_bib}

\appendix
\section{Additional Preliminaries} \label{app:additionalBackground}

\subsection{\INDqCPA Security of Public Key Encryption Schemes} \label{app:additionalBackgroundPKE}

The security game for \INDqCPA security~\cite{C:BonZha13} of public key encryption schemes is defined as follows.
We note that this notion is equivalent to the standard QS1 security notion for public key encryption schemes.
\begin{exper}\label{exp:PKEINDqCPA}
	The {\em \INDqCPA experiment} $\gamefont{IND\mbox{-}qCPA}(\Sigma, \Eve, \secpar)$ for a PKE scheme $\Sigma=(\KGen,\Enc,\Dec)$ is defined as follows:
	\begin{algorithmic}[1]
		\State \chal runs $(\pk,\sk) \from \KGen$
		\State $\Eve(\pk) \to (\mssg[0],\mssg[1], \sigma_{\varfont{state}})$
		\State \chal receives $\mssg[0], \mssg[1]$ and does the following:
		\begin{itemize}
			\item[-] flips $b \fromunif \bit$
			\item[-] samples $\rnd \fromunif \rndsp$
			\item[-] computes $\Enc_\pk(\mssg[b];\rnd) \to \ctxt$
			\item[-] sends \ctxt to \Eve
		\end{itemize}
		\State $\Eve(\sigma_{\varfont{state}},\ctxt) \to b' \in \bit$
		\State \textbf{if} $b = b'$ \textbf{then return} \win; \textbf{else return} \rej.
	\end{algorithmic}
\end{exper}
Security is defined as negligible advantage over guessing in winning the security game.
\begin{definition}[\INDqCPA, PKE]\label{def:PKEINDqCPA}
	A PKE scheme $\Sigma$ has {\em ciphertext indistinguishability under quantum chosen plaintext attack}, or it is \INDqCPA-secure, iff for any QPT adversary \Eve it holds:
	\begin{align*}
		\left| \Pr\left[ \gamefont{IND\mbox{-}qCPA}(\Sigma, \Eve, \secpar) \to \win \right]- \frac{1}{2}\right| \leq \negl(\secpar)\,.
	\end{align*}
\end{definition}

\subsection{Symmetric Key Encryption} \label{app:additionalBackgroundSKE}

Below we define symmetric key encryption (SKE) schemes.
\begin{definition}
	A symmetric key encryption \emph{(SKE)} scheme $\Sigma$ is a tuple of three efficient algorithms $(\KGen,\Enc,\Dec)$ such that:
	\begin{itemize}
		\item $\KGen \colon \NN \rightarrow \keysp$ is the (randomised) encryption algorithm which takes a security parameter \secpar as input, and returns a key \key.
		
		\item $\Enc \colon \keysp \times \mssgsp \rightarrow \ctxtsp$ is the (randomised) encryption algorithm which takes a key \key and a message \mssg as input, and returns a ciphertext \ctxt.
		
		\item $\Dec \colon \keysp \times \ctxtsp \rightarrow \mssgsp$ is the decryption algorithm which takes as input a key \key and a ciphertext \ctxt, and returns a message \mssg.
	\end{itemize}
	By \keysp, \mssgsp, and \ctxtsp, we denote the key space, message space, and ciphertext space, respectively.
\end{definition}
Next, we define the security game for \qIND security, following~\cite{C:GagHulSch16}.
\begin{exper}\label{exp:SKEqINDqCPA}
	The {\em \qIND experiment} $\gamefont{qIND}(\Sigma, \Eve, \secpar)$ for an SKE scheme $\Sigma=(\KGen,\Enc,\Dec)$ is defined as follows:
	\begin{algorithmic}[1]
		\State \chal runs $\key \getsr \KGen(\secpar)$ and implements $O^{(2)}_\Enc$
		\State $\Eve() \to (\phi_0, \phi_1, \sigma_{\varfont{state}})$
		\State \chal receives $\phi_0, \phi_1$ and does the following: 
		\begin{itemize}
			\item[-] flips $b \fromunif \bit$
			\item[-] traces out $\phi_{1-b}$
			\item[-] calls $\psi \from O^{(2)}_\Enc(\phi_b)$
			\item[-] sends $\psi$ to \Eve
		\end{itemize}
		\State $\Eve(\sigma_{\varfont{state}},\psi) \to b' \in \bit$
		\State \textbf{if} $b = b'$ \textbf{then return} \win; \textbf{else return} \rej.
	\end{algorithmic}
\end{exper}
Just as for our new security notion, security is defined as negligible advantage over guessing in winning the game.
\begin{definition}[\qIND, SKE]\label{def:SKEqIND}
	An SKE scheme $\Sigma$ has {\em quantum ciphertext indistinguishability}, or it is \qIND-secure, iff for any QPT adversary \Eve it holds:
	\begin{align*}
		\left| \Pr\left[ \gamefont{qIND}(\Sigma, \Eve, \secpar) \to \win \right]- \frac{1}{2}\right| \leq \negl(\secpar)\,.
	\end{align*}
\end{definition}

\section{The Role of Randomness Superposition}\label{app:randsuperposition}

In this section we discuss the possibility of having superposition of randomness in the type-2 challenge query.
So far, we have only considered the case of classical randomness, as this is chosen by the (honest) challenger.
But one could consider scenarios where the adversary can somehow trick the challenger into using a superposition of 
randomness in the challenge query.
Here we discuss two possible ways to deal with this issue, one of which turns out to be unachievable while the other 
yields a notion equivalent to the one we propose in Section~\ref{sec:QuantumSecurityForPKE}. 

Assume that the challenger chooses a superposition of randomness to encrypt one of the messages chosen by the adversary.
Following our security experiment, the challenger would keep the randomness register and merely send the ciphertext 
register to the adversary.
The crucial observation is that the registers containing the randomness and the ciphertext are now entangled.
As observed in~\cite{C:GagHulSch16}, withholding the randomness register is equivalent to measuring 
it from the point of view of the adversary.
This means that this approach would in fact be equivalent to our security notion using a classical randomness.

Alternatively, to prevent the aforementioned issue of entanglement between the challenger and the adversary, we might let the challenger send the randomness register to the adversary.
However, the resulting security notion is unachievable as it would allow the adversary to always distinguish 
encryptions.
We illustrate this with the following attack.
First, the adversary chooses two distinct classical messages \mssg[0], \mssg[1], and executes the \qIND challenge query with these two.
Upon receiving the ciphertext register and the randomness register, the adversary evaluates (locally) the type-1 encryption operator initialising the input register with \ket{\mssg[0]}, the randomness register with the randomness state received from the challenger, and the ancilla register with the received ciphertext.
Finally, the adversary measures the ciphertext register output of the type-1 encryption operator: if he measures $0$, 
then he outputs $b=0$, otherwise outputs $b=1$.
The circuit is depicted in Fig.~\ref{fig:superpositionRandomnessAttack}.
The attack works because, if $b=0$, then the adversary will compute the same ciphertext as the challenger, hence the 
output register of the type-1 encryption will be $\ket{0}$; on the other hand, if $b=1$, a random value will be observed 
instead.
Clearly, this results in output states that the adversary can distinguish with overwhelming probability.

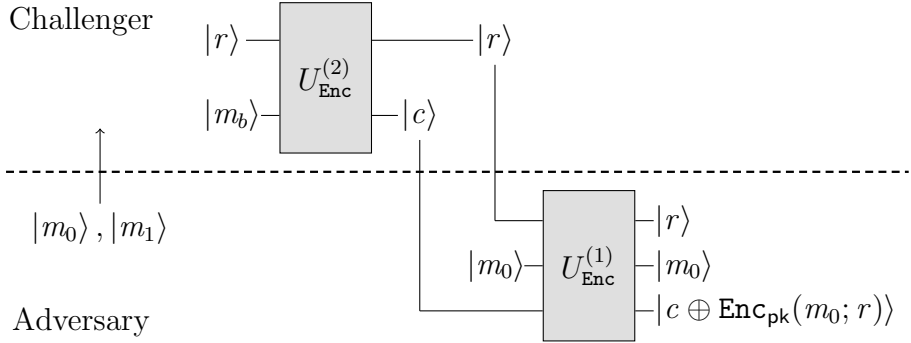
\begin{figure}[htbp]
	\centering
	\scalebox{\scaleFactorGraphics}{
		\begin{tikzpicture}
			
			\node (MIDDLE) [text width=1em,align=center] {};

			\node (UencChallenger)[text width=1.0cm,align=center,fill=gray!25,minimum height=2.0cm,minimum width=1.0cm,draw] [above of=MIDDLE, xshift = -1.75cm, node distance = 1.25cm] {\large $U_{\Enc}^{(2)}$};
			
			\node (UencAdversary)[text width=1.0cm,align=center,fill=gray!25,minimum height=2.0cm,minimum width=1.0cm,draw] [below of=MIDDLE, xshift = 1.75cm, node distance = 1.25cm] {\large $U_{\Enc}^{(1)}$};
			
			\node (ChallengeMessages)[text width=2.0cm,align=center] [left of=UencAdversary, node distance = 6.5cm, yshift = 0.5cm] {\large $\ket{\mssg[0]},\ket{\mssg[1]}$};
			\draw[->] ($(ChallengeMessages.north) + (0.0,0.0)$) -- ($(ChallengeMessages.north) + (0.0,1.0)$);

			\node (INCHALLENGER1) [text width=0.75cm,align=left] [left of=UencChallenger, yshift=0.5cm, node distance=1.25cm] {\large \ket{\rnd}};
			\node (INCHALLENGER2) [text width=0.75cm,align=center] [left of=UencChallenger, yshift=-0.5cm, node distance=1.25cm] {\large \ket{\mssg[b]}};
			
			\node (OUTCHALLENGER1) [text width=0.75cm,align=center] [right of=UencChallenger, yshift=0.5cm, node distance=2.25cm] {\large \ket{\rnd}};
			\node (OUTCHALLENGER2) [text width=0.75cm,align=center] [right of=UencChallenger, yshift=-0.5cm, node distance=1.25cm] {\large \ket{\ctxt}};
			
			\node (INADVERSARY1) [text width=0.75cm,align=center] [left of=UencAdversary, yshift=0.6cm, node distance=1.25cm] {};
			\node (INADVERSARY2) [text width=0.75cm,align=center] [left of=UencAdversary, yshift=0.0cm, node distance=1.25cm] {\large \ket{\mssg[0]}};
			\node (INADVERSARY3) [text width=0.75cm,align=center] [left of=UencAdversary, yshift=-0.6cm, node distance=2.25cm] {};
			
			\node (OUTADVERSARY1) [text width=0.75cm,align=left] [right of=UencAdversary, yshift=0.6cm, node distance=1.25cm] {\large \ket{\rnd}};
			\node (OUTADVERSARY2) [text width=0.75cm,align=center] [right of=UencAdversary, yshift=0.0cm, node distance=1.25cm] {\large \ket{\mssg[0]}};
			\node (OUTADVERSARY3) [text width=0.75cm,align=center] [right of=UencAdversary, yshift=-0.6cm, node distance=1.25cm] {\large \ket{\ctxt \xor \enc_{\pk}(\mssg[0];\rnd)}};
			
			\draw[] ($(INCHALLENGER1.east) + (-0.3,0.0)$) -- ($(UencChallenger.west) + (0.0,0.5)$);
			\draw[] ($(INCHALLENGER2.east) + (-0.1,0.0)$) -- ($(UencChallenger.west) + (0.0,-0.5)$);
			
			\draw[] ($(UencChallenger.east) + (0.0,0.5)$) -- ($(OUTCHALLENGER1.west) + (0.2,0.0)$);
			\draw[] ($(UencChallenger.east) + (0.0,-0.5)$) -- ($(OUTCHALLENGER2.west) + (0.2,0.0)$);
			
			\draw[] ($(OUTCHALLENGER1.south) + (0.0,0.0)$) -- ($(INADVERSARY1) + (0.0,0.0)$);
			\draw[] ($(OUTCHALLENGER2.south) + (0.0,0.0)$) -- ($(INADVERSARY3) + (0.0,0.0)$);
			
			\draw[] ($(INADVERSARY1) + (0.0,0.0)$) -- ($(UencAdversary.west) + (0.0,0.6)$);
			\draw[] ($(INADVERSARY2.east) + (-0.1,0.0)$) -- ($(UencAdversary.west) + (0.0,0.0)$);
			\draw[] ($(INADVERSARY3) + (0.0,0.0)$) -- ($(UencAdversary.west) + (0.0,-0.6)$);
			
			\draw[] ($(UencAdversary.east) + (0.0,0.6)$) -- ($(OUTADVERSARY1.west) + (0.1,0.0)$);
			\draw[] ($(UencAdversary.east) + (0.0,0.0)$) -- ($(OUTADVERSARY2.west) + (0.1,0.0)$);
			\draw[] ($(UencAdversary.east) + (0.0,-0.6)$) -- ($(OUTADVERSARY3.west) + (0.1,0.0)$);

			\draw[densely dashed, line width = 0.3mm] ($(MIDDLE) + (-6.0,0.0)$) -- ($(MIDDLE) + (6.0,0.0)$);
			\node (CHALLENGER)[text width=2.0cm,align=center,] [above of=MIDDLE, node distance = 2.0cm, xshift = -5.0cm] {\large Challenger};
			\node (ADVERSARY)[text width=2.0cm,align=center,] [below of=MIDDLE, node distance = 2.0cm, xshift = -5.0cm] {\large Adversary};
		\end{tikzpicture}
	}
	\caption{Generic attack against superposition of randomness.}
	\label{fig:superpositionRandomnessAttack}
\end{figure}

\section{Concurrent Work} \label{app:concurrentWork}

In concurrent and independent work, Chevalier et al.~\cite{eprint:ChevalierEV20} and Carstens et al.~\cite{eprint:CETU20} propose alternative QS2 security notions for public and symmetric key encryption schemes.
There are important, conceptual differences between these works and ours which we illustrate in this section.

Chevalier et al. start by resuming a game-based quantum indistinguishability notion previously introduced by Mossayebi and Schack~\cite{MossayebiS16} which, we conjecture, is not comparable to ours.
This notion is based on a real-or-permuted approach: in the security game, the adversary sends a {\em single} quantum plaintext of the form $\sum_x \alpha_x \ket{x}$ and (depending on the value of the secret challenge bit $b$) receives back either $\sum_x \alpha_x \ket{x, \Enc(x)}$, or $\sum_x \alpha_x \ket{x, \Enc(\pi(x))}$, where $\pi$ is a random permutation implemented by the challenger. 
To avoid confusion with our notion (\qINDqCPA), we refer to their notion as \piqINDqCPA.
Consider the canonical \INDCPA symmetric key encryption scheme that works by XOR-ing the message with $\PRF_{\key}(\rnd)$, where \PRF is a keyed pseudorandom function and \rnd is a freshly sampled randomness which is then attached to the resulting ciphertext. This scheme was previously known to be secure according to Boneh and Zhandry's \INDqCPA notion; however, in~\cite{MossayebiS16}, Mossayebi and Schack show that such scheme is not \piqINDqCPA secure,\footnote{The proof of the attack is only sketched, whereas it is formally given by Chevalier et al. in~\cite{eprint:ChevalierEV20}.} thereby yielding a separation result.

Starting from this consideration, the authors of~\cite{eprint:ChevalierEV20} develop a framework of new QS2 security notions (both for the symmetric and public key case) where the challenge query is quantum but implemented as a single message in the real-or-permuted setting.
This approach has advantages and disadvantages compared with the one in~\cite{C:GagHulSch16} (for the symmetric key case) and the one we adopt in this work (for the public key case):

\begin{itemize}
	\item The notion of \piqINDqCPA (and the related \CCA and non-malleability notions) only require the use of type-1 oracles, therefore greatly simplifying the modelling of the security game. Another advantage is that it can be defined for any encryption scheme, while we require {\em isometric schemes} (cf. Section~\ref{sec:classif}).
	
	\item On the other hand, the notion of Chevalier et al. (unlike ours) deviates from the established framework for the classical case. In the traditional setting of symmetric and public key security notions, in fact, it is well-known that many different characterisations of \INDCPA (with two or more messages chosen by the adversary, with one chosen and one random or fixed, etc.) are all equivalent to an intuitive (but more cumbersome) notion of semantic security. For \piqINDqCPA, however, it is {\em crucial} that the adversary can only send {\em one single challenge message} to the challenger.\footnote{Chevalier et al. prove the {\em composability} of their notions, but this refers to the fact that one can formulate their security game using multiple challenge queries, where each query is still restricted to a \emph{single} message.}
	This is not the case for our \qINDqCPA (and related) notions: although we do not write them down explicitly here (we leave them for a future update in the appendix of this manuscript), all these good `sanity checks' can be easily inferred by:
	\begin{itemize}
		\item The lifting from a two-message \qIND (QS2) challenge query to a two-message \QIND (QS3) challenge query as shown in~\cite{DBLP:phd/dnb/Gagliardoni17};
		\item The equivalence between different types of \QIND challenge query (two messages, many messages, real-or-random, etc.) as shown in~\cite{C:BroJef15};
		\item The equivalence of such \QIND notion to a sound notion of {\em quantum semantic security} as from~\cite{ICITS:ABFGSJ16}.
	\end{itemize}
	This means that our notions (in the public key case) and the ones in~\cite{C:GagHulSch16} (for the symmetric key case) closely mirror the well-established framework in the classical setting.
	
	\item Analogously, because of the presence of entanglement between plaintext and ciphertext registers, the notions by Chevalier et al. do not mirror the existing solid framework for {\em fully quantum notions} (QS3 setting) in the literature. This is not a flaw by itself, but it has the drawback that many useful tools cannot be straightforwardly `imported' from the QS3 setting. An example mentioned above is the difficulty of formalising the equivalence of \piqINDqCPA to a natural notion of quantum semantic security, or the possibility of easily lifting the QS2 security of a classical scheme $\Sigma$ to the QS3 security of a quantum scheme $\Pi$ that uses $\Sigma$ as a building block. Another example is the difficulty of defining quantum \CCAtwo security, which can be done in a relatively easy way in the QS3 setting with the real-vs-ideal approach by Alagic et al. from~\cite{EC:AlaGagMaj18}, while requiring the more involved compressed oracle technique by Zhandry~\cite{C:Zhandry19} for the results in~\cite{eprint:ChevalierEV20}.\footnote{The authors of~\cite{eprint:ChevalierEV20} also explain in their work why Alagic et al.'s approach would not work in their case.}
	
	\item Chevalier et al. expand substantially Mossayebi and Schack's results, answering many questions left previously open (some of which also mentioned in an early version of the present work) such as the security of the encrypt-then-MAC construction.
	Moreover they introduce a technique (based on Zhandry's compressed oracles) to record queries and simulate answers to inverse oracles which is of independent interest.
	
	\item It is important to notice that the separation result by Chevalier et al. and Mossayeby and Schack rely on entanglement between message and ciphertext register and not on a particular weakness in the scheme. 
	In contrast, our separation (and the one in~\cite{C:GagHulSch16}) relies solely on a property of the encryption scheme in question. 
	One has to consider how the ability of a quantum adversary of receiving back an entangled pair of message and ciphertext really mirrors the classical intuition, where an adversary would only receive a ciphertext instead.
	
	\item Finally, and most importantly, in the present work we show that for many real-world PKE schemes (including most of the NIST candidates) type-2 encryption operators can be implemented {\em without knowledge of the secret key}. 
	This invalidates Chevalier et al.'s argument that type-2 operators are unreasonable in the public key setting, and actually makes the need for our \qINDqCPA notion in the public key case stronger than ever.
\end{itemize}
Ultimately, we think that the contribution of Chevalier et al. is of great importance and their results are undoubtedly interesting.
It is important to notice that the canonical \INDCPA scheme used by Chevalier et al. and Mossayebi and Schack as a separation from Boneh and Zhandry's \INDqCPA is also shown to be insecure according to the \qINDqCPA security notion in~\cite{C:GagHulSch16}. We can hence see \piqINDqCPA as a QS2 security notion which is incomparable to the \qINDqCPA notion we present in this work, with advantages and disadvantages as explained above.

In a recent work, Carstens et al.~\cite{eprint:CETU20} study in detail the relationships between existing security notions for QS2 encryption.
Their work is mainly focused on SKE, and they prove certain separations based on (reasonable) conjectures.
In particular their work supports our conjecture that \qINDqCPA and \piqINDqCPA are incomparable also in the PKE case.

\end{document}